\documentclass[aps,prb,notitlepage,nofootinbib,floatfix,twocolumn,superscriptaddress, longbibliography]{revtex4-2}
\usepackage{dsfont}
\usepackage{amsmath,amsthm}
\usepackage{amssymb}
\usepackage{amsfonts}
\usepackage{mathtools}
\usepackage[tight]{subfigure}
\usepackage{enumitem}
\usepackage{soul}
\usepackage{cancel}
\usepackage{multirow}
\usepackage{tikz}
\usepackage{pifont} 
\usepackage{xspace}
\usepackage{comment}
\usepackage{physics}
\usepackage{graphicx}
\usepackage{dcolumn}
\usepackage{bm}
\usepackage[breaklinks]{hyperref}
\hypersetup{colorlinks=true, linkcolor=blue, citecolor=blue, filecolor=blue, urlcolor=blue}
\usepackage{ulem}
\usepackage[dvipsnames]{xcolor}

\theoremstyle{definition}

\newtheorem*{definition*}{Definition} 
\newtheorem{theorem}{Theorem}
\newtheorem*{theorem*}{Theorem}

\newtheorem*{lemma*}{Lemma}
\newtheorem*{proposition*}{Proposition}

\begin{document}

\title{State $k$-designs from Hamiltonian evolution}

\author{Shengxian Hou}
\affiliation{School of Physics, Peking University, Beijing 100871, China}

\author{Zong-Yue Hou}
\affiliation{School of Physics, Peking University, Beijing 100871, China}

\author{Zhi-Cheng Yang}
\thanks{Corresponding author}
\email{zcyang19@pku.edu.cn}
\affiliation{School of Physics, Peking University, Beijing 100871, China}
\affiliation{Center for High Energy Physics, Peking University, Beijing 100871, China}

\date{\today}

\begin{abstract}

We study the generation of state \(k\)-designs from time evolution under a fixed Hamiltonian. Specifically, we consider the ensemble $\mathcal{E}=\left\{e^{-iHt}\ket{\psi_0}\ \middle|\ t\sim \mathrm{Unif}[0,T],\ \ket{\psi_0}\sim \mathcal{E}'\right\}$, where the initial states are sampled from an ensemble \(\mathcal{E}'\). For Hamiltonians drawn from the Gaussian unitary ensemble, we derive a simple relation between the frame potential of the evolved ensemble \(\mathcal{E}\) and that of the initial ensemble \(\mathcal{E}'\) in the large evolution time limit. This relation shows that \(\mathcal{E}\) forms an exact state \(k\)-design in the thermodynamic limit as long as \(\mathcal{E}'\) forms a state 1-design. Remarkably, we further show, both analytically and numerically, that time evolution under a simple nonintegrable mixed-field Ising Hamiltonian can generate approximate state \(k\)-designs with high precision, starting from product states in an appropriately chosen Pauli basis. We also analyze the finite-$T$ correction and find it scales as $O(1/T)$. To reduce the evolution time, we propose an $M$-step quench protocol that suppresses this correction to $O(1/T^M)$, which is also verified numerically. We then extend our analysis to unitary ensembles, deriving an analogous recursion relation for the unitary frame potential. Our results elucidate the mechanisms underlying recent proposals for generating unitary \(k\)-designs through sequential quantum quenches in a unified manner.
  
\end{abstract}

\maketitle

\textit{Introduction.--} A foundational idea of statistical mechanics is that deterministic dynamics can generate statistical uniformity. In classical Hamiltonian systems, an initially localized distribution in phase space does not literally relax to a smooth equilibrium distribution at the fine-grained level, since phase-space volume is preserved by Liouville's theorem. Instead, under chaotic mixing, it is stretched and folded into increasingly fine structures, so that coarse-grained observations become indistinguishable from sampling the accessible energy shell. This dynamical route to equilibrium underlies the replacement of long-time dynamical averages by phase-space ensemble averages~\cite{Arnold1968,Cornfeld1982}.

Quantum mechanics sharpens this question. The space of pure states is the projective Hilbert space, equipped with the unitarily invariant Haar measure, which plays the role of the equal-weight ensemble over pure quantum states. Haar-random states therefore provide the quantum analogue of an ideal equal-weight state ensemble~\cite{Haar1933}. However, exact Haar randomness is an excessively strong requirement. Physical observables and many information-theoretic protocols usually depend only on finitely many copies of the state, or equivalently on finite moments of the state distribution. This motivates the notion of a quantum state \(k\)-design: an ensemble of pure states whose \(k\)-copy moments agree with those of the Haar ensemble. State designs therefore provide a finite-resolution notion of Hilbert-space uniformity, in which the design order \(k\) quantifies how deeply an ensemble mimics Haar randomness~\cite{Renes_2004,Scott_2008,Ambainis2007}.

This perspective suggests a natural form of quantum mixing: starting from simple states and evolving under a chaotic Hamiltonian, can one generate an ensemble that is indistinguishable from Haar randomness to finite moment order? Besides probing a form of quantum ergodicity stronger than ordinary thermalization, dynamically generated state designs are useful resources for tasks ranging from randomized benchmarking to shadow tomography~\cite{Knill_2008,Huang_2020}.

Two complementary approaches to this question have emerged. In the projected-ensemble setting, measurements on part of an evolved many-body state induce an ensemble on the remaining subsystem that can form high-order state designs, a phenomenon known as deep thermalization~\cite{Ho_2022,Ippoliti_2023,PRXQuantum.4.010311,Choi_2023,Mok_2025}. Temporal ensembles instead draw their randomness from the evolution time of a fixed Hamiltonian~\cite{PhysRevX.14.041051,Ghosh_2025}. In this setting, energy conservation prevents the orbit of a generic initial state from reproducing the Haar ensemble~\cite{PhysRevX.14.041051,mok2026naturestingyuniversalityscrooge}. Although Haar-like features of time-evolved states have been investigated through entanglement statistics and related observables~\cite{Ghosh_2025,ghosh2025randomizationtimesquantumchaotic}, the conditions under which temporal ensembles form state designs remain less understood. Parallel questions have also been studied for unitary ensembles~\cite{Zhou_2026,Zhou2026,Sun2026,Cui2025}.

\begin{figure*}[!t]
    \centering
    \includegraphics[width=\textwidth]{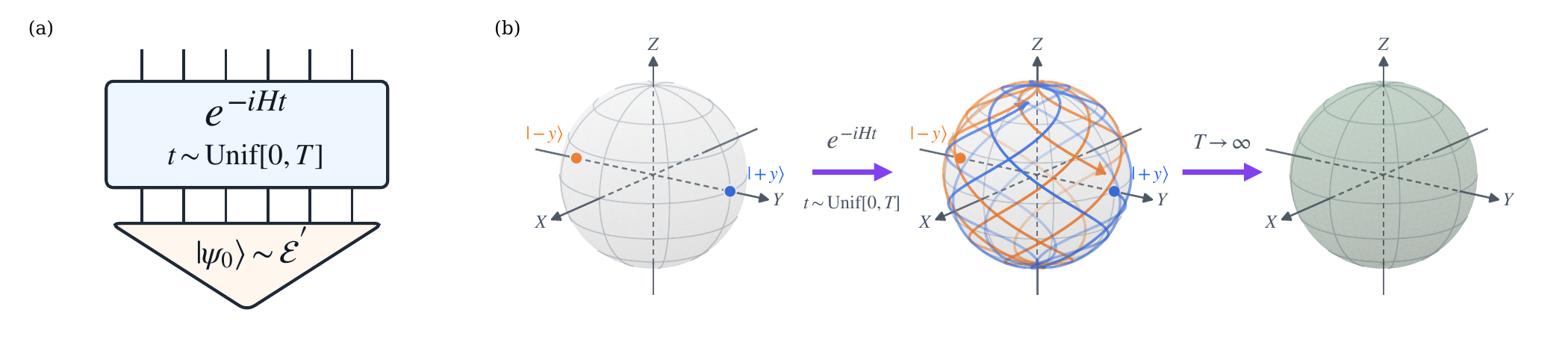}
    \caption{(a) The setup considered in this work, $\mathcal{E}=\{ e^{-iHt}\ket{\psi _{0}} \mid t\sim \mathrm{Unif}[0,T],\ket{\psi _{0}} \sim \mathcal{E}' \}$. (b) Illustration of the formation of state designs under Hamiltonian evolution. Each initial state in $\mathcal{E}'$ generates a trajectory in the Hilbert space visualized on a Bloch sphere. In the long-time limit, these trajectories cover the Hilbert space uniformly and the evolved ensemble $\mathcal{E}$ forms a state $k$-design.}
    \label{fig:demonstration}
\end{figure*}

In this Letter, we resolve this question by constructing temporal ensembles from suitably chosen initial-state ensembles. Specifically, we study $\mathcal{E} =\{e^{-iHt}\ket{\psi_0}\ |\ t\sim \mathrm{Unif}[0,T],\ \ket{\psi_0}\sim \mathcal{E}'\}$, where the initial state $\ket{\psi_0}$ is sampled from an ensemble $\mathcal{E}'$ [see Fig.~\ref{fig:demonstration}(a)]. This framework is highly general: $\mathcal{E}'$ may consist of a single state or may itself be a temporal ensemble generated by evolution under another Hamiltonian. For a fixed Hamiltonian $H$ sampled from the Gaussian unitary ensemble (GUE), we derive a simple relation between the frame potentials of $\mathcal{E}$ and $\mathcal{E}'$ in the long-time limit. This relation shows that choosing $\mathcal{E}'$ to be a state $1$-design is sufficient for $\mathcal{E}$ to approach a state $k$-design in the thermodynamic limit. Remarkably, we further demonstrate, both analytically and numerically, that evolution under a simple nonintegrable mixed-field Ising Hamiltonian generates approximate state $k$-designs with high precision when initialized in product states drawn from an appropriately chosen Pauli basis. A detailed analysis of the temporal convergence shows that the leading finite-$T$ correction scales as $O(T^{-1})$. To accelerate convergence, we propose an $M$-step quench protocol involving two mixed-field Ising Hamiltonians, which suppresses this correction to $O(T^{-M})$, as is confirmed numerically. Finally, we extend our analysis to unitary ensembles and derive an analogous recursion relation for the unitary frame potential, providing a unified explanation of several recent constructions of unitary $k$-designs from sequential quantum quenches~\cite{Zhou2026,Sun2026}.

\textit{Frame potential.--} In this Letter, we use $D$ to denote the Hilbert space dimension and $N$ to denote system size. We quantify the proximity of an ensemble (state or unitary) to the Haar ensemble using frame potentials. The $k$-th  frame potential of a state ensemble $\mathcal{E}$ is defined as $F^{(k)}_{\mathcal{E}}:= \mathbb{E}_{\ket{\psi},\ket{\phi}\sim \mathcal{E}}|\braket{\psi}{\phi}|^{2k}$, which is minimized by the Haar state ensemble with $F^{(k)}_{\mathrm{Haar}}=k!/[D(D+1)\cdots(D+k-1)]=k!/D^k+O(1/D^{k+1})$. Similarly, the $k$-th  frame potential of a unitary ensemble $\mathcal{E}$ is defined as $\mathcal{F}^{(k)}_{\mathcal{E}}:=\mathbb{E}_{U,V\sim \mathcal{E}}|\mathrm{tr}[U^{\dagger}V]|^{2k}$, which is minimized by the Haar unitary ensemble with $\mathcal{F}^{(k)}_{\mathrm{Haar}}=k!$ for $D\geq k$. A state (unitary) ensemble forms a state (unitary) $k$-design if and only if its $k$-th frame potential coincides with the corresponding Haar value. Throughout this Letter we work in the regime $D\gg k$. The relations between frame potentials and the additive errors of approximate designs are listed in End Matter.

\textit{State ensembles.--} We study state ensembles obtained by evolving an initial-state ensemble under a fixed time-independent Hamiltonian. Specifically, we consider the ensemble
\begin{equation}
    \mathcal{E}=\{ e^{-iHt}\ket{\psi _{0}} \mid t\sim \mathrm{Unif}[0,T],\ket{\psi _{0}} \sim \mathcal{E}' \},
    \label{defstate}
\end{equation}
where $t$ is uniformly sampled from the time interval $[0,T]$, and the initial state $\ket{\psi _{0}}$ is sampled from a state ensemble $\mathcal{E}'$. We will assume the eigenenergies of $H$ satisfy the $k$-th no-resonance condition: any two subsets of eigenenergies $\{ E_{n_{1}},\cdots,E_{n_{k}} \}$ and $\{E_{m_{1}},\cdots,E_{m_{k}}  \}$ satisfy $\sum _{a=1}^{k}E_{n_{a}}=\sum _{a=1}^{k}E_{m_{a}}$ if and only if there exists a permutation $\sigma \in S_{k}$ such that $n_{a}=m_{\sigma(a)}$ for $a=1,2,\cdots,k$. Such no-resonance conditions are expected to hold  for generic chaotic Hamiltonians~\cite{10.21468/SciPostPhys.15.4.165,huang2020instabilitylocalizationtranslationinvariantsystems}. For now, we work in the long time limit $T\to\infty$.

If the initial ensemble $\mathcal{E}'$ consists of a single state $\ket{\Psi _{0}}$, then $\mathcal{E}$ is called a temporal ensemble~\cite{PhysRevX.14.041051}.  It has been proved that if  $\ket{\Psi_{0}}=\frac{1}{\sqrt{D}}\sum _{n}\ket{E_{n}}$, i.e., the initial state has a uniform amplitude on all eigenstates, $\mathcal{E}$ forms an approximate state $k$-design with additive error $O(k^{2}/D)$~\cite{Nakata_2014}. However, for  generic initial states, temporal ensembles form approximate Scrooge $k$-designs with density matrix $\sigma _{\mathrm{diag}}=\sum _{n}|\langle E_{n}|\Psi_{0}\rangle|^{2}\ket{E_{n}}\bra{E_{n}}$ instead of Haar state $k$-designs~\cite{mok2026naturestingyuniversalityscrooge}.

In this work, we show that  drawing initial states from a simple ensemble enables the formation of state $k$-designs. Assuming the $k$-th no-resonance condition and taking $T\to \infty$, the frame potential of $\mathcal{E}$ is 
\begin{equation}
    F^{(k)}_{\mathcal{E}}=k!\sum _{\vec{n}}\frac{1}{\Omega(\vec{n})}\left( \bra{E_{n_{1}}\cdots E_{n_{k}}} \rho ^{(k)}_{\mathcal{E}'}\ket{E_{n_{1}}\cdots E_{n_{k}}}  \right) ^{2},
    \label{statefp}
\end{equation}
where $\vec{n}\in \{1,\dots,D\}^{\otimes k}$, $\rho^{(k)}_{\mathcal{E}'}:={\mathbb{E}}_{|\psi\rangle\sim\mathcal{E}'}[\left(|\psi\rangle\langle\psi|\right)^{\otimes k}]$ is the \(k\)-th moment operator, and $\Omega(\vec{n})=\prod _{i=1}^{D}r_{i}!$, where $r_i=\sum_{a=1}^k \delta_{in_a}$ is the multiplicity of $i$ in $\vec{n}$. This term arises from using the no-resonance condition to eliminate the dependence on eigenenergies. A detailed derivation is given in the Supplemental Material (SM)~\cite{SM}. For Hamiltonians drawn from GUE, we prove the following result:

\begin{theorem}
    Consider the state ensemble $\mathcal{E}$ defined in Eq.~\eqref{defstate}. Assuming  $H$ satisfies the $k$-th no-resonance condition and taking the limit $T\to\infty$, the GUE-averaged frame potential is
    \begin{equation}
        \mathop{\mathbb{E}}\limits_{H\sim \mathrm{GUE}}F^{(k)}_{\mathcal{E}}=\frac{k!}{D^k}\left(1+\sum_{m=1}^k \binom{k}{m}F^{(m)}_{\mathcal{E}'}\right)+\text{subleading}.
        \label{stategue}
    \end{equation}
    The subleading term is at least $O(1/D)$ smaller than the leading term. The proof is given in the SM~\cite{SM}.
    \label{theostate}
\end{theorem}

Theorem~\ref{theostate} implies that $\mathcal{E}$ forms a state $k$-design if the first-order frame potential of the initial ensemble $\mathcal{E}'$ satisfies $F^{(1)}_{\mathcal{E}'}=O(1/D)$. This follows from the fact that the state frame potential is nonincreasing with respect to its order $k$. Therefore as long as $F^{(1)}_{\mathcal{E}'}=O(1/D)$, the leading order of ${\mathbb{E}}_{H\sim \mathrm{GUE}}F^{(k)}_{\mathcal{E}}$ will be $k!/D^k$, matching the leading order of the frame potential of the Haar ensemble. In particular, we can choose $\mathcal{E}'$ to be a state 1-design, since its first-order frame potential is $1/D$. The simplest choice of a state 1-design is a complete orthonormal basis, e.g. computational basis states. Another simple construction is  randomly applying Pauli strings to an arbitrary fixed initial reference state $\ket{0}$, which gives $\mathcal{E}'=\{P\ket{0}\mid P\sim \mathcal{P}_N\}$, where $\mathcal{P}_N=\{I,X,Y,Z\}^{\otimes N}$ is the set of  $N$-qubit Pauli strings. This follows from the fact that $\mathcal{P}_N$ forms a unitary 1-design.

Theorem~\ref{theostate} can be viewed as a recursion relation for the frame potential, since the evolved ensemble $\mathcal{E}$ can in turn be used as the initial ensemble for a second Hamiltonian evolution. This recursive structure immediately implies that two independent GUE Hamiltonian quenches are sufficient to generate state $k$-designs, starting from a \textit{fixed} reference state.  First, taking the  ensemble $\mathcal{E}'$ to consist only of a single state $\ket{\psi_0}$, we obtain $\mathcal{E}_1=\{e^{-iH_1t_1}\ket{\psi_0}\mid t_1\sim\mathrm{Unif}[0,T]\}$. Using Theorem~\ref{theostate}, we get
\begin{equation}
    \mathop{\mathbb{E}}\limits_{H_1\sim \mathrm{GUE}}F^{(k)}_{\mathcal{E}_1}=\frac{2^k k!}{D^k}+\text{subleading},
\end{equation}
which is not even a 1-design, consistent with the previous result that generic temporal ensembles are not state $k$-designs, but we have $\mathop{\mathbb{E}}\limits_{H_1\sim \mathrm{GUE}}F^{(1)}_{\mathcal{E}_1}=2/D=O(1/D)$. Then we apply the second quench to get $\mathcal{E}_2=\{e^{-iH_2t_2}\ket{\psi}\mid  t_2\sim\mathrm{Unif}[0,T],\ket{\psi}\sim \mathcal{E}_1 \}$. Using Theorem~\ref{theostate} again gives
\begin{equation}
    \mathop{\mathbb{E}}\limits_{H_1,H_2\sim \mathrm{GUE}}F^{(k)}_{\mathcal{E}_2}=\frac{k!}{D^k}+\text{subleading},
\end{equation}
matching the leading order of the Haar ensemble.

\textit{Local Hamiltonians.--} While the GUE provides a convenient starting point for obtaining analytical results, such Hamiltonians are highly nonlocal and thus physically unrealistic. We now ask whether a fixed, physically realistic chaotic Hamiltonian can generate state $k$-designs when the initial ensemble is chosen appropriately. As a first step, let us take $\mathcal{E}'$ to be an orthonormal product-state basis $\mathcal{E}'=\{\ket{m}\}_{m=1}^{D}$. For qubit systems, simple examples are the product bases formed by bitstrings in the $X$-, $Y$-, or $Z$-basis. Substituting $\rho^{(k)}_{\mathcal{E}'}=D^{-1}\sum_{m=1}^{D}\ket{m}\bra{m}^{\otimes k}$ into Eq.~\eqref{statefp}, we obtain
\begin{equation}
    F^{(k)}_{\mathcal{E}}=\frac{k!}{D^{2}}\sum _{m,m'}\sum _{\vec{n}}\frac{1}{\Omega(\vec{n})}\prod _{a=1}^{k} |  \langle E_{n_{a}}|m\rangle |^{2}|\langle E_{n_{a}}|m'\rangle|^{2}.
\end{equation}
We show in the SM that if the overlaps $x_{nm}:=|\langle E_n|m\rangle|^2$ obey the Porter-Thomas distribution and may be treated as independent random variables, then the leading contribution is precisely the Haar value, $F^{(k)}_{\mathcal{E}}=k!/D^k$~\cite{PT}. Porter-Thomas statistics thus provide a simple sufficient condition for design formation. Nevertheless, for local Hamiltonians, eigenstates near the spectral edges are typically highly structured and can deviate strongly from Porter-Thomas statistics. Moreover, different states in $\mathcal{E}'$ generally have different energy expectation values and therefore place their spectral weight in different regions of the many-body spectrum. Consequently, for a generic choice of product-state basis, the resulting ensemble $\mathcal{E}$ need not approach a state $k$-design under evolution by a local Hamiltonian~\cite{Cui2025}. In particular, we prove the following no-go theorem:
\begin{theorem}
    Consider a state ensemble $\mathcal{E}=\{e^{-iHt}\ket{\psi_0}\mid t\sim\mathcal{D},\ket{\psi_0}\sim\mathcal{E}'\}$, where $\mathcal{D}$ is an arbitrary distribution of evolution times. Suppose $\mathrm{tr}(H)=0$ and expand $H$ in the Pauli basis as $H=\sum_{P\in\mathcal{P}_N}h_PP$. Then $\mathcal{E}$ cannot be an approximate state $2$-design with additive error
    \begin{equation}
        \epsilon<\frac{\bigg|\mathop{\mathbb{E}}\limits_{\ket{\psi_0} \sim \mathcal{E}'}\bra{\psi_0} H\ket{\psi_0} ^{2}-\frac{1}{D+1}\sum _{P\in \mathcal{P}_{N}}h_{P}^{2}\bigg|}{\lVert H\rVert_{\infty}^{2}}.
    \end{equation}
    The proof is given in the SM~\cite{SM}.
    \label{theonogo}
\end{theorem}

We apply Theorem~\ref{theonogo} to the mixed-field Ising Hamiltonian with open boundary conditions,
\begin{equation}
    H=J\sum_{i=1}^{N-1} Z_iZ_{i+1}+h_x\sum_{i=1}^N X_i+h_z\sum_{i=1}^N Z_i.
    \label{eq:MFIM}
\end{equation}
For $\mathcal{E}'$ consisting of bitstrings in the $Z$-basis, we have $\mathbb{E}_{\ket{z}\sim\mathcal{E}'}\langle z|H|z\rangle^2=\sum_{P\in\{I,Z\}^{\otimes N}}h_P^2$, which is not small. Similarly for bitstrings in the $X$-basis. Nevertheless, product states in the $Y$-basis are not ruled out by Theorem~\ref{theonogo}. We therefore take $\mathcal{E}'$ to be the $Y$-basis bitstrings. In Fig.~\ref{fig:N_scaling}(a), we show that for fixed nonintegrable parameters in Eq.~\eqref{eq:MFIM}, the resulting long-time ensemble shows excellent agreement with Haar random ensemble, for $k$ up to 5. Moreover, we demonstrate in Fig.~\ref{fig:N_scaling}(b) that the relative error $\delta F^{(k)}_{\mathcal{E}}:=(F^{(k)}_{\mathcal{E}}-F^{(k)}_{\mathrm{Haar}})/F^{(k)}_{\mathrm{Haar}}$ decreases exponentially with system size.

The special role of the $Y$-basis for the mixed-field Ising model can be understood intuitively by noting $\langle y|H|y\rangle=0$ and $\langle y|H^2|y\rangle=D^{-1}\operatorname{tr}(H^2)$, which hold for every $Y$-basis bitstring $\ket{y}$. Moreover, as shown in the SM, the energy population of any such state converges to that of the infinite-temperature ensemble~\cite{SM}. By eigenstate thermalization~\cite{Srednicki1994, Rigol2008, DAlessio2016}, any product state in the $Y$-basis thermalizes to infinite temperature under Hamiltonian~\eqref{eq:MFIM}. Therefore, the $Y$-basis states predominantly sample highly excited eigenstates in the bulk of the spectrum, where chaotic behavior and random-matrix-like eigenstate statistics are expected~\cite{atas2015quantumisingmodeltransverse}.

\begin{figure}[!t]
    \centering
    \includegraphics[width=\linewidth]{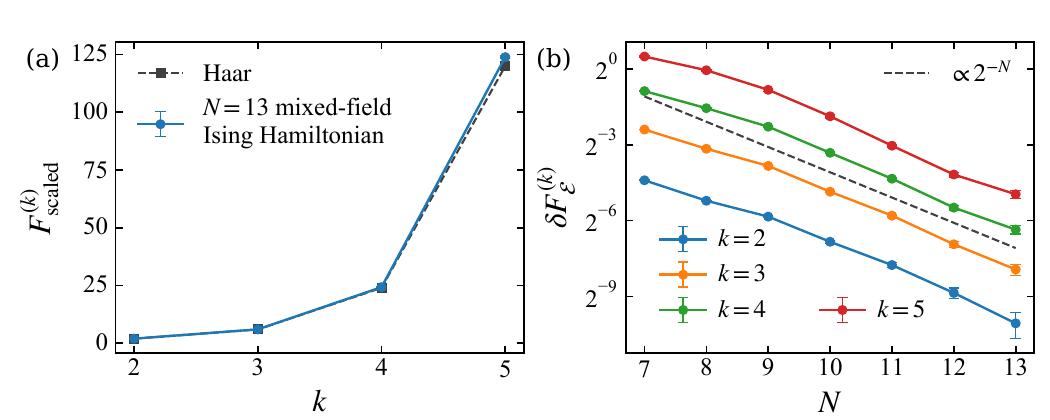}
    \caption{(a) The scaled frame potential $F^{(k)}_{\mathrm{scaled}}:=D(D+1)\cdots(D+k-1)F^{(k)}_{\mathcal{E}}$ for the ensemble $\mathcal{E}=\{e^{-iHt}\ket{y}\mid t\sim\mathrm{Unif}[0,T],\ket{y}\sim Y\text{-basis} \}$ in the limit $T\to\infty$, where $H$ is the mixed-field Ising Hamiltonian~\eqref{eq:MFIM} with parameters $(J,h_x,h_z)=(1,0.9045,0.809)$ and system size $N=13$. The black dots indicate the Haar values $k!$. (b) The relative error compared to the Haar ensemble $\delta F^{(k)}_{\mathcal{E}}:= (F^{(k)}_{\mathcal{E}}-F^{(k)}_{\mathrm{Haar}})/F^{(k)}_{\mathrm{Haar}}$ for different $N$ and $k$, which shows exponential convergence upon increasing system sizes.} 
    \label{fig:N_scaling}
\end{figure}

\textit{Temporal convergence.--} We now analyze the convergence to state $k$-designs at large but finite evolution times. We begin by rewriting the frame potential of the ensemble defined in Eq.~\eqref{defstate} as
\begin{equation}
    F^{(k)}_{\mathcal{E}}=\frac{2}{T}\int ^{T}_{0}d\tau\left( 1-\frac{\tau}{T} \right)f ^{(k)}_{\mathcal{E}'}(\tau),
\end{equation}
where $f ^{(k)}_{\mathcal{E}'}(\tau)=\mathbb{E}_{y,y'\sim\mathcal{E}' }|\bra{y}e^{iH\tau}\ket{y'}|^{2k}$. For the $Y$-basis initial ensemble, one has $f^{(k)}_{\mathcal{E}'}(0)=1/D$. At short times, the dominant contribution comes from $y=y'$. Expanding to order $\tau^2$, we have $f^{(k)}_{\mathcal{E}'}(\tau)\approx \frac{1}{D}e^{-k\langle H^2\rangle \tau^2}$, with a characteristic decay timescale $\tau^*\sim 1/\sqrt{N}$. The leading finite-$T$ correction to $F^{(k)}_{\mathcal{E}}$ is therefore controlled by the short-time region of the integral and scales as $O(\frac{\tau ^{*}}{DT})$. Since $F^{(k)}_{\rm Haar}$ is of order $D^{-k}$ and $\tau^*$ scales at most polynomially with $N$, this correction becomes subleading when $T=\Omega(D^{k})$, scaling exponentially with system size. The predicted $T^{-1}$ convergence is confirmed numerically in Fig.~\ref{fig:T_scaling} for $k=3$ and $k=4$. For more general initial ensemble $\mathcal{E}'$, the leading finite-$T$ correction scales as $O(F^{(k)}_{\mathcal{E}'}\tau ^{*}/T)$. In the SM, we establish this scaling explicitly for GUE Hamiltonians~\cite{SM}.

To reduce the evolution time, we consider a protocol with two Hamiltonians $H_{1}$ and $H_{2}$ applied alternately. An $M$-step quench yields the ensemble $\mathcal{E}_{M}=\{ e^{-iH_{1/2}t_{M}}\cdots e^{-iH_{1}t_{3}}e^{-iH_{2}t_{2}}e^{-iH_{1}t_{1}}\ket{\psi}\mid t_{1},\cdots,t_{M}\sim \mathrm{Unif}[0,T],\ket{\psi}\sim \mathcal{E}' \}$ where the Hamiltonian in the last step depends on whether $M$ is even or odd. According to the previous argument, the finite-$T$ correction to this ensemble is expected to scale as $O(F^{(k)}_{\mathcal{E}'}/T^{M})$, and we confirm this prediction numerically. In Fig.~\ref{fig:T_scaling}, we choose two mixed-field Ising Hamiltonians with different parameters to construct one-, two- and three-step quench ensembles. The relative error decreases with $T$ precisely as predicted above, scaling as $T^{-M}$ for $M$ steps.  To make the total evolution time linear in system size $N$, we can choose $M=N$ and $T=2^k$, such that the total evolution time is $MT=2^kN$.

\begin{figure}[!t]
    \centering
    \includegraphics[width=\linewidth]{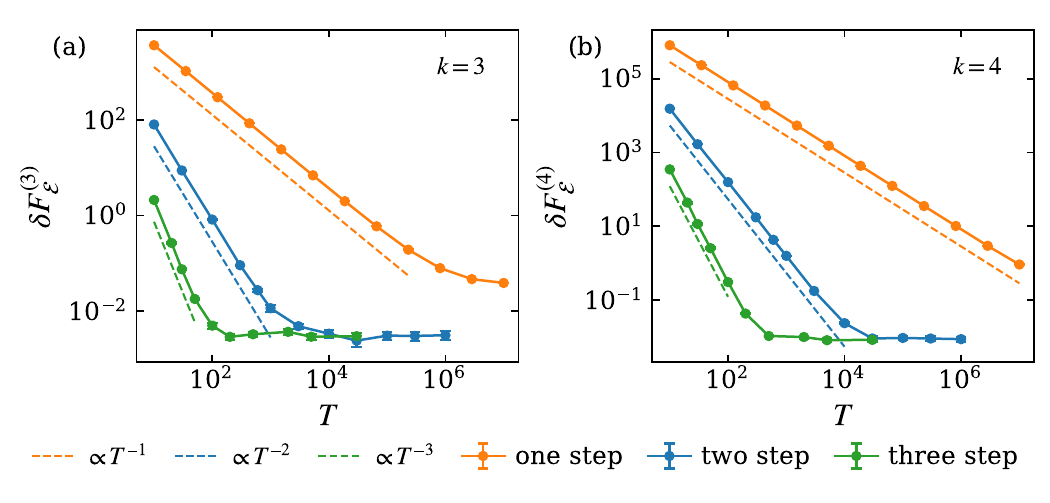}
    \caption{The relative error compared to the Haar ensemble $\delta F^{(k)}_{\mathcal{E}}$ at different values of $T$ for one-, two- and three-step quenches for $k=3,4$. The initial ensemble $\mathcal{E}'$ is chosen as the $Y$-basis. For multiple quenches, we use two mixed-field Ising Hamiltonians~\eqref{eq:MFIM} with parameters $(J,h_x,h_z)=(1,0.9045,0.809)$ for $H_1$ and $(J,h_x,h_z)=(1,-1.05,0.5)$ for $H_2$. We take system size $N=10$. The three dashed lines are proportional to $T^{-1}$, $T^{-2}$ and $T^{-3}$ respectively.}
    \label{fig:T_scaling}
\end{figure}

\textit{Unitary ensembles.--} We now turn to the generation of unitary $k$-designs from Hamiltonian evolution. Analogously, we consider an initial unitary ensemble $\mathcal{E}'$ composed with  Hamiltonian evolution acting either on the left or on the right. This yields two ensembles,
\begin{equation}
    \mathcal{E}_L=\{e^{-iHt}V\mid t\sim\mathrm{Unif}[0,T], V\sim \mathcal{E}' \}
    \label{uleft}
\end{equation}
and
\begin{equation}
    \mathcal{E}_R=\{Ve^{-iHt}\mid t\sim\mathrm{Unif}[0,T], V\sim \mathcal{E}' \}.
    \label{uright}
\end{equation}
As in the state-ensemble case, we derive a recursion relation for the unitary frame potential, for $H$ drawn from GUE.
\begin{theorem}
    Consider unitary ensembles $\mathcal{E}_L$ and $\mathcal{E}_R$ defined in Eq.~\eqref{uleft} and Eq.~\eqref{uright}. Assuming the Hamiltonian $H$ satisfies the $k$-th no-resonance condition and in the limit $T\to\infty$, the GUE-averaged frame potential is
    \begin{equation}
        \mathop{\mathbb{E}}\limits_{H\sim \mathrm{GUE}}\mathcal{F}^{(k)}_{\mathcal{E}_{L/R}}=k!\left( 1+\sum _{m=1}^{k}\binom{k}{m}\frac{\mathcal{F}^{(m)}_{\mathcal{E}'}}{D^{m}} \right)+\text{subleading},
    \end{equation}
    where $L/R$ indicates the formula holds for both $\mathcal{E}_L$  and $\mathcal{E}_R$. The subleading term is at least $O(1/D)$ smaller than the leading term. The proof is given in the SM~\cite{SM}.
    \label{theounitary}
\end{theorem}

Since for unitary ensembles, $\mathcal{F}^{(m)}_{\mathcal{E}'}/D^m$ is not necessarily  non-increasing with respect to the order $m$, the first-order frame potential of $\mathcal{E}'$ alone does not determine whether $\mathcal{E}_{L/R}$ forms unitary $k$-designs. Higher order information of $\mathcal{E}'$ is therefore required. 

We can use Theorem~\ref{theounitary} to explain recent proposals for generating unitary $k$-designs from Hamiltonian evolution. Ref.~\cite{Zhou2026} shows that three GUE Hamiltonians are sufficient to generate unitary $k$-designs, whereas two Hamiltonians are not. More specifically, the authors constructed unitary ensembles $\mathcal{E}_{2\mathrm{SP}}=\{e^{-iH_2t_2}e^{-iH_1t_1}\mid t_1,t_2\sim\mathrm{Unif}[0,T]\}$ and $\mathcal{E}_{3\mathrm{SP}}=\{e^{-iH_3t_3}e^{-iH_2t_2}e^{-iH_1t_1}\mid t_1,t_2,t_3\sim\mathrm{Unif}[0,T]\}$. We can now immediately see this by using Theorem~\ref{theounitary} recursively. First taking the ensemble  $\mathcal{E}'$ to consist only of the identity operator, we obtain $\mathcal{E}_{1\mathrm{SP}}=\{e^{-iH_1t_1}\mid t_1\sim \mathrm{Unif}[0,T]\}$. Using Theorem~\ref{theounitary} we get
\begin{equation}
    \mathop{\mathbb{E}}\limits_{H_1\sim \mathrm{GUE}}\mathcal{F}^{(k)}_{\mathcal{E}_{1\mathrm{SP}}}=k!D^k+\text{subleading}.
\end{equation}
Applying a second quench to $\mathcal{E}_{1\mathrm{SP}}$ will give us $\mathcal{E}_{2\mathrm{SP}}$. Using Theorem~\ref{theounitary} again, we get
\begin{equation}
    \mathop{\mathbb{E}}\limits_{H_1,H_2\sim \mathrm{GUE}}\mathcal{F}^{(k)}_{\mathcal{E}_{2\mathrm{SP}}}=k!\left( 1+\sum _{m=1}^{k}\binom{k}{m}m!  \right)+\text{subleading},
\end{equation}
which coincides with the results in Ref.~\cite{Zhou2026}. Note that the frame potential of $\mathcal{E}_{2\mathrm{SP}}$ does not equal the Haar value $k!$, but it is already of order $O(1)$. Therefore, by using Theorem~\ref{theounitary} a third time, we know that $\mathcal{E}_{3\mathrm{SP}}$ forms a unitary $k$-design. 

Another construction of unitary $k$-designs was proposed in Ref.~\cite{Sun2026}, where the authors construct the ensemble $\{e^{-iH_2t_2}Pe^{-iH_1t_1}\mid t_1,t_2\sim \mathrm{Unif}[0,T],P\sim \mathcal{P}_N\}$ and show that it forms a unitary $k$-design for GUE Hamiltonians. This result also follows directly from Theorem~\ref{theounitary}. Moreover, Theorem~\ref{theounitary} suggests that the particular location of Pauli string insertion does not matter.

\textit{Discussions.--} We have shown that temporal randomness generated by a fixed Hamiltonian can be promoted to finite-moment Haar randomness by suitably choosing the initial-state ensemble. For GUE Hamiltonians, we derive a recursion relation connecting the frame potentials before and after the evolution, which shows that any initial state $1$-design is sufficient to generate state $k$-designs in the thermodynamic limit. We further demonstrate that the same mechanism works for a local nonintegrable mixed-field Ising Hamiltonian for initial product states in an appropriately chosen Pauli basis inspired by eigenstate thermalization. We show that the leading finite-time correction scales as $T^{-1}$ and can be suppressed to $T^{-M}$ using an $M$-step alternating-quench protocol, which allows for convergence to $k$-designs in polynomial time. Finally, we extend our framework to unitary designs, revealing the common mechanism underlying recent constructions of unitary $k$-designs from Hamiltonian evolution. An important open question is to determine general criteria under which a local Hamiltonian can generate state designs, as well as a prescription for constructing the appropriate initial ensemble. It would also be interesting to establish whether fixed local Hamiltonians can generate unitary $k$-designs.

\textit{Acknowledgments.-} This work is supported by Grant No. 12375027 from the National Natural Science Foundation of China. Numerical simulations were performed on the High-performance Computing Platform of Peking University.

\bibliographystyle{apsrev4-2}
\bibliography{ref}

\begin{thebibliography}{40}%
\makeatletter
\providecommand \@ifxundefined [1]{%
 \@ifx{#1\undefined}
}%
\providecommand \@ifnum [1]{%
 \ifnum #1\expandafter \@firstoftwo
 \else \expandafter \@secondoftwo
 \fi
}%
\providecommand \@ifx [1]{%
 \ifx #1\expandafter \@firstoftwo
 \else \expandafter \@secondoftwo
 \fi
}%
\providecommand \natexlab [1]{#1}%
\providecommand \enquote  [1]{``#1''}%
\providecommand \bibnamefont  [1]{#1}%
\providecommand \bibfnamefont [1]{#1}%
\providecommand \citenamefont [1]{#1}%
\providecommand \href@noop [0]{\@secondoftwo}%
\providecommand \href [0]{\begingroup \@sanitize@url \@href}%
\providecommand \@href[1]{\@@startlink{#1}\@@href}%
\providecommand \@@href[1]{\endgroup#1\@@endlink}%
\providecommand \@sanitize@url [0]{\catcode `\\12\catcode `\$12\catcode `\&12\catcode `\#12\catcode `\^12\catcode `\_12\catcode `\%12\relax}%
\providecommand \@@startlink[1]{}%
\providecommand \@@endlink[0]{}%
\providecommand \url  [0]{\begingroup\@sanitize@url \@url }%
\providecommand \@url [1]{\endgroup\@href {#1}{\urlprefix }}%
\providecommand \urlprefix  [0]{URL }%
\providecommand \Eprint [0]{\href }%
\providecommand \doibase [0]{https://doi.org/}%
\providecommand \selectlanguage [0]{\@gobble}%
\providecommand \bibinfo  [0]{\@secondoftwo}%
\providecommand \bibfield  [0]{\@secondoftwo}%
\providecommand \translation [1]{[#1]}%
\providecommand \BibitemOpen [0]{}%
\providecommand \bibitemStop [0]{}%
\providecommand \bibitemNoStop [0]{.\EOS\space}%
\providecommand \EOS [0]{\spacefactor3000\relax}%
\providecommand \BibitemShut  [1]{\csname bibitem#1\endcsname}%
\let\auto@bib@innerbib\@empty
\bibitem [{\citenamefont {Arnold}\ and\ \citenamefont {Avez}(1968)}]{Arnold1968}%
  \BibitemOpen
  \bibfield  {author} {\bibinfo {author} {\bibfnamefont {V.~I.}\ \bibnamefont {Arnold}}\ and\ \bibinfo {author} {\bibfnamefont {A.}~\bibnamefont {Avez}},\ }\href@noop {} {\emph {\bibinfo {title} {Ergodic Problems of Classical Mechanics}}}\ (\bibinfo  {publisher} {W.~A.~Benjamin},\ \bibinfo {address} {New York},\ \bibinfo {year} {1968})\BibitemShut {NoStop}%
\bibitem [{\citenamefont {Cornfeld}\ \emph {et~al.}(1982)\citenamefont {Cornfeld}, \citenamefont {Fomin},\ and\ \citenamefont {Sinai}}]{Cornfeld1982}%
  \BibitemOpen
  \bibfield  {author} {\bibinfo {author} {\bibfnamefont {I.~P.}\ \bibnamefont {Cornfeld}}, \bibinfo {author} {\bibfnamefont {S.~V.}\ \bibnamefont {Fomin}},\ and\ \bibinfo {author} {\bibfnamefont {Y.~G.}\ \bibnamefont {Sinai}},\ }\href@noop {} {\emph {\bibinfo {title} {Ergodic Theory}}},\ \bibinfo {series} {Grundlehren der mathematischen Wissenschaften}, Vol.\ \bibinfo {volume} {245}\ (\bibinfo  {publisher} {Springer},\ \bibinfo {address} {New York},\ \bibinfo {year} {1982})\BibitemShut {NoStop}%
\bibitem [{\citenamefont {Haar}(1933)}]{Haar1933}%
  \BibitemOpen
  \bibfield  {author} {\bibinfo {author} {\bibfnamefont {A.}~\bibnamefont {Haar}},\ }\href@noop {} {\bibfield  {journal} {\bibinfo  {journal} {Annals of Mathematics}\ }\textbf {\bibinfo {volume} {34}},\ \bibinfo {pages} {147} (\bibinfo {year} {1933})}\BibitemShut {NoStop}%
\bibitem [{\citenamefont {Renes}\ \emph {et~al.}(2004)\citenamefont {Renes}, \citenamefont {Blume-Kohout}, \citenamefont {Scott},\ and\ \citenamefont {Caves}}]{Renes_2004}%
  \BibitemOpen
  \bibfield  {author} {\bibinfo {author} {\bibfnamefont {J.~M.}\ \bibnamefont {Renes}}, \bibinfo {author} {\bibfnamefont {R.}~\bibnamefont {Blume-Kohout}}, \bibinfo {author} {\bibfnamefont {A.~J.}\ \bibnamefont {Scott}},\ and\ \bibinfo {author} {\bibfnamefont {C.~M.}\ \bibnamefont {Caves}},\ }\href {https://doi.org/10.1063/1.1737053} {\bibfield  {journal} {\bibinfo  {journal} {Journal of Mathematical Physics}\ }\textbf {\bibinfo {volume} {45}},\ \bibinfo {pages} {2171–2180} (\bibinfo {year} {2004})}\BibitemShut {NoStop}%
\bibitem [{\citenamefont {Scott}(2008)}]{Scott_2008}%
  \BibitemOpen
  \bibfield  {author} {\bibinfo {author} {\bibfnamefont {A.~J.}\ \bibnamefont {Scott}},\ }\href {https://doi.org/10.1088/1751-8113/41/5/055308} {\bibfield  {journal} {\bibinfo  {journal} {Journal of Physics A: Mathematical and Theoretical}\ }\textbf {\bibinfo {volume} {41}},\ \bibinfo {pages} {055308} (\bibinfo {year} {2008})}\BibitemShut {NoStop}%
\bibitem [{\citenamefont {Ambainis}\ and\ \citenamefont {Emerson}(2007)}]{Ambainis2007}%
  \BibitemOpen
  \bibfield  {author} {\bibinfo {author} {\bibfnamefont {A.}~\bibnamefont {Ambainis}}\ and\ \bibinfo {author} {\bibfnamefont {J.}~\bibnamefont {Emerson}},\ }\href@noop {} {\bibinfo {title} {Quantum t-designs: t-wise independence in the quantum world}} (\bibinfo {year} {2007}),\ \Eprint {https://arxiv.org/abs/quant-ph/0701126} {arXiv:quant-ph/0701126} \BibitemShut {NoStop}%
\bibitem [{\citenamefont {Knill}\ \emph {et~al.}(2008)\citenamefont {Knill}, \citenamefont {Leibfried}, \citenamefont {Reichle}, \citenamefont {Britton}, \citenamefont {Blakestad}, \citenamefont {Jost}, \citenamefont {Langer}, \citenamefont {Ozeri}, \citenamefont {Seidelin},\ and\ \citenamefont {Wineland}}]{Knill_2008}%
  \BibitemOpen
  \bibfield  {author} {\bibinfo {author} {\bibfnamefont {E.}~\bibnamefont {Knill}}, \bibinfo {author} {\bibfnamefont {D.}~\bibnamefont {Leibfried}}, \bibinfo {author} {\bibfnamefont {R.}~\bibnamefont {Reichle}}, \bibinfo {author} {\bibfnamefont {J.}~\bibnamefont {Britton}}, \bibinfo {author} {\bibfnamefont {R.~B.}\ \bibnamefont {Blakestad}}, \bibinfo {author} {\bibfnamefont {J.~D.}\ \bibnamefont {Jost}}, \bibinfo {author} {\bibfnamefont {C.}~\bibnamefont {Langer}}, \bibinfo {author} {\bibfnamefont {R.}~\bibnamefont {Ozeri}}, \bibinfo {author} {\bibfnamefont {S.}~\bibnamefont {Seidelin}},\ and\ \bibinfo {author} {\bibfnamefont {D.~J.}\ \bibnamefont {Wineland}},\ }\bibfield  {journal} {\bibinfo  {journal} {Physical Review A}\ }\textbf {\bibinfo {volume} {77}},\ \href {https://doi.org/10.1103/physreva.77.012307} {10.1103/physreva.77.012307} (\bibinfo {year} {2008})\BibitemShut {NoStop}%
\bibitem [{\citenamefont {Huang}\ \emph {et~al.}(2020)\citenamefont {Huang}, \citenamefont {Kueng},\ and\ \citenamefont {Preskill}}]{Huang_2020}%
  \BibitemOpen
  \bibfield  {author} {\bibinfo {author} {\bibfnamefont {H.-Y.}\ \bibnamefont {Huang}}, \bibinfo {author} {\bibfnamefont {R.}~\bibnamefont {Kueng}},\ and\ \bibinfo {author} {\bibfnamefont {J.}~\bibnamefont {Preskill}},\ }\href {https://doi.org/10.1038/s41567-020-0932-7} {\bibfield  {journal} {\bibinfo  {journal} {Nature Physics}\ }\textbf {\bibinfo {volume} {16}},\ \bibinfo {pages} {1050–1057} (\bibinfo {year} {2020})}\BibitemShut {NoStop}%
\bibitem [{\citenamefont {Ho}\ and\ \citenamefont {Choi}(2022)}]{Ho_2022}%
  \BibitemOpen
  \bibfield  {author} {\bibinfo {author} {\bibfnamefont {W.~W.}\ \bibnamefont {Ho}}\ and\ \bibinfo {author} {\bibfnamefont {S.}~\bibnamefont {Choi}},\ }\bibfield  {journal} {\bibinfo  {journal} {Physical Review Letters}\ }\textbf {\bibinfo {volume} {128}},\ \href {https://doi.org/10.1103/physrevlett.128.060601} {10.1103/physrevlett.128.060601} (\bibinfo {year} {2022})\BibitemShut {NoStop}%
\bibitem [{\citenamefont {Ippoliti}\ and\ \citenamefont {Ho}(2023)}]{Ippoliti_2023}%
  \BibitemOpen
  \bibfield  {author} {\bibinfo {author} {\bibfnamefont {M.}~\bibnamefont {Ippoliti}}\ and\ \bibinfo {author} {\bibfnamefont {W.~W.}\ \bibnamefont {Ho}},\ }\bibfield  {journal} {\bibinfo  {journal} {PRX Quantum}\ }\textbf {\bibinfo {volume} {4}},\ \href {https://doi.org/10.1103/prxquantum.4.030322} {10.1103/prxquantum.4.030322} (\bibinfo {year} {2023})\BibitemShut {NoStop}%
\bibitem [{\citenamefont {Cotler}\ \emph {et~al.}(2023)\citenamefont {Cotler}, \citenamefont {Mark}, \citenamefont {Huang}, \citenamefont {Hern\'andez}, \citenamefont {Choi}, \citenamefont {Shaw}, \citenamefont {Endres},\ and\ \citenamefont {Choi}}]{PRXQuantum.4.010311}%
  \BibitemOpen
  \bibfield  {author} {\bibinfo {author} {\bibfnamefont {J.~S.}\ \bibnamefont {Cotler}}, \bibinfo {author} {\bibfnamefont {D.~K.}\ \bibnamefont {Mark}}, \bibinfo {author} {\bibfnamefont {H.-Y.}\ \bibnamefont {Huang}}, \bibinfo {author} {\bibfnamefont {F.}~\bibnamefont {Hern\'andez}}, \bibinfo {author} {\bibfnamefont {J.}~\bibnamefont {Choi}}, \bibinfo {author} {\bibfnamefont {A.~L.}\ \bibnamefont {Shaw}}, \bibinfo {author} {\bibfnamefont {M.}~\bibnamefont {Endres}},\ and\ \bibinfo {author} {\bibfnamefont {S.}~\bibnamefont {Choi}},\ }\href {https://doi.org/10.1103/PRXQuantum.4.010311} {\bibfield  {journal} {\bibinfo  {journal} {PRX Quantum}\ }\textbf {\bibinfo {volume} {4}},\ \bibinfo {pages} {010311} (\bibinfo {year} {2023})}\BibitemShut {NoStop}%
\bibitem [{\citenamefont {Choi}\ \emph {et~al.}(2023)\citenamefont {Choi}, \citenamefont {Shaw}, \citenamefont {Madjarov}, \citenamefont {Xie}, \citenamefont {Finkelstein}, \citenamefont {Covey}, \citenamefont {Cotler}, \citenamefont {Mark}, \citenamefont {Huang}, \citenamefont {Kale}, \citenamefont {Pichler}, \citenamefont {Brandão}, \citenamefont {Choi},\ and\ \citenamefont {Endres}}]{Choi_2023}%
  \BibitemOpen
  \bibfield  {author} {\bibinfo {author} {\bibfnamefont {J.}~\bibnamefont {Choi}}, \bibinfo {author} {\bibfnamefont {A.~L.}\ \bibnamefont {Shaw}}, \bibinfo {author} {\bibfnamefont {I.~S.}\ \bibnamefont {Madjarov}}, \bibinfo {author} {\bibfnamefont {X.}~\bibnamefont {Xie}}, \bibinfo {author} {\bibfnamefont {R.}~\bibnamefont {Finkelstein}}, \bibinfo {author} {\bibfnamefont {J.~P.}\ \bibnamefont {Covey}}, \bibinfo {author} {\bibfnamefont {J.~S.}\ \bibnamefont {Cotler}}, \bibinfo {author} {\bibfnamefont {D.~K.}\ \bibnamefont {Mark}}, \bibinfo {author} {\bibfnamefont {H.-Y.}\ \bibnamefont {Huang}}, \bibinfo {author} {\bibfnamefont {A.}~\bibnamefont {Kale}}, \bibinfo {author} {\bibfnamefont {H.}~\bibnamefont {Pichler}}, \bibinfo {author} {\bibfnamefont {F.~G. S.~L.}\ \bibnamefont {Brandão}}, \bibinfo {author} {\bibfnamefont {S.}~\bibnamefont {Choi}},\ and\ \bibinfo {author} {\bibfnamefont {M.}~\bibnamefont {Endres}},\ }\href {https://doi.org/10.1038/s41586-022-05442-1} {\bibfield  {journal} {\bibinfo  {journal} {Nature}\ }\textbf {\bibinfo {volume} {613}},\ \bibinfo {pages} {468–473} (\bibinfo {year} {2023})}\BibitemShut {NoStop}%
\bibitem [{\citenamefont {Mok}\ \emph {et~al.}(2025)\citenamefont {Mok}, \citenamefont {Haug}, \citenamefont {Shaw}, \citenamefont {Endres},\ and\ \citenamefont {Preskill}}]{Mok_2025}%
  \BibitemOpen
  \bibfield  {author} {\bibinfo {author} {\bibfnamefont {W.-K.}\ \bibnamefont {Mok}}, \bibinfo {author} {\bibfnamefont {T.}~\bibnamefont {Haug}}, \bibinfo {author} {\bibfnamefont {A.~L.}\ \bibnamefont {Shaw}}, \bibinfo {author} {\bibfnamefont {M.}~\bibnamefont {Endres}},\ and\ \bibinfo {author} {\bibfnamefont {J.}~\bibnamefont {Preskill}},\ }\bibfield  {journal} {\bibinfo  {journal} {Physical Review Letters}\ }\textbf {\bibinfo {volume} {134}},\ \href {https://doi.org/10.1103/physrevlett.134.180403} {10.1103/physrevlett.134.180403} (\bibinfo {year} {2025})\BibitemShut {NoStop}%
\bibitem [{\citenamefont {Mark}\ \emph {et~al.}(2024)\citenamefont {Mark}, \citenamefont {Surace}, \citenamefont {Elben}, \citenamefont {Shaw}, \citenamefont {Choi}, \citenamefont {Refael}, \citenamefont {Endres},\ and\ \citenamefont {Choi}}]{PhysRevX.14.041051}%
  \BibitemOpen
  \bibfield  {author} {\bibinfo {author} {\bibfnamefont {D.~K.}\ \bibnamefont {Mark}}, \bibinfo {author} {\bibfnamefont {F.}~\bibnamefont {Surace}}, \bibinfo {author} {\bibfnamefont {A.}~\bibnamefont {Elben}}, \bibinfo {author} {\bibfnamefont {A.~L.}\ \bibnamefont {Shaw}}, \bibinfo {author} {\bibfnamefont {J.}~\bibnamefont {Choi}}, \bibinfo {author} {\bibfnamefont {G.}~\bibnamefont {Refael}}, \bibinfo {author} {\bibfnamefont {M.}~\bibnamefont {Endres}},\ and\ \bibinfo {author} {\bibfnamefont {S.}~\bibnamefont {Choi}},\ }\href {https://doi.org/10.1103/PhysRevX.14.041051} {\bibfield  {journal} {\bibinfo  {journal} {Phys. Rev. X}\ }\textbf {\bibinfo {volume} {14}},\ \bibinfo {pages} {041051} (\bibinfo {year} {2024})}\BibitemShut {NoStop}%
\bibitem [{\citenamefont {Ghosh}\ \emph {et~al.}(2025{\natexlab{a}})\citenamefont {Ghosh}, \citenamefont {Langlett}, \citenamefont {Hunter-Jones},\ and\ \citenamefont {Rodriguez-Nieva}}]{Ghosh_2025}%
  \BibitemOpen
  \bibfield  {author} {\bibinfo {author} {\bibfnamefont {S.}~\bibnamefont {Ghosh}}, \bibinfo {author} {\bibfnamefont {C.~M.}\ \bibnamefont {Langlett}}, \bibinfo {author} {\bibfnamefont {N.}~\bibnamefont {Hunter-Jones}},\ and\ \bibinfo {author} {\bibfnamefont {J.~F.}\ \bibnamefont {Rodriguez-Nieva}},\ }\bibfield  {journal} {\bibinfo  {journal} {Physical Review B}\ }\textbf {\bibinfo {volume} {112}},\ \href {https://doi.org/10.1103/bwr6-vskn} {10.1103/bwr6-vskn} (\bibinfo {year} {2025}{\natexlab{a}})\BibitemShut {NoStop}%
\bibitem [{\citenamefont {Mok}\ \emph {et~al.}(2026)\citenamefont {Mok}, \citenamefont {Haug}, \citenamefont {Ho},\ and\ \citenamefont {Preskill}}]{mok2026naturestingyuniversalityscrooge}%
  \BibitemOpen
  \bibfield  {author} {\bibinfo {author} {\bibfnamefont {W.-K.}\ \bibnamefont {Mok}}, \bibinfo {author} {\bibfnamefont {T.}~\bibnamefont {Haug}}, \bibinfo {author} {\bibfnamefont {W.~W.}\ \bibnamefont {Ho}},\ and\ \bibinfo {author} {\bibfnamefont {J.}~\bibnamefont {Preskill}},\ }\href {https://arxiv.org/abs/2601.00266} {\bibinfo {title} {Nature is stingy: Universality of scrooge ensembles in quantum many-body systems}} (\bibinfo {year} {2026}),\ \Eprint {https://arxiv.org/abs/2601.00266} {arXiv:2601.00266 [quant-ph]} \BibitemShut {NoStop}%
\bibitem [{\citenamefont {Ghosh}\ \emph {et~al.}(2025{\natexlab{b}})\citenamefont {Ghosh}, \citenamefont {Hunter-Jones},\ and\ \citenamefont {Rodriguez-Nieva}}]{ghosh2025randomizationtimesquantumchaotic}%
  \BibitemOpen
  \bibfield  {author} {\bibinfo {author} {\bibfnamefont {S.}~\bibnamefont {Ghosh}}, \bibinfo {author} {\bibfnamefont {N.}~\bibnamefont {Hunter-Jones}},\ and\ \bibinfo {author} {\bibfnamefont {J.~F.}\ \bibnamefont {Rodriguez-Nieva}},\ }\href {https://arxiv.org/abs/2512.25074} {\bibinfo {title} {Randomization times under quantum chaotic hamiltonian evolution}} (\bibinfo {year} {2025}{\natexlab{b}}),\ \Eprint {https://arxiv.org/abs/2512.25074} {arXiv:2512.25074 [cond-mat.stat-mech]} \BibitemShut {NoStop}%
\bibitem [{\citenamefont {Zhou}\ \emph {et~al.}(2026{\natexlab{a}})\citenamefont {Zhou}, \citenamefont {Löwenberg},\ and\ \citenamefont {Sonner}}]{Zhou_2026}%
  \BibitemOpen
  \bibfield  {author} {\bibinfo {author} {\bibfnamefont {Y.-N.}\ \bibnamefont {Zhou}}, \bibinfo {author} {\bibfnamefont {R.}~\bibnamefont {Löwenberg}},\ and\ \bibinfo {author} {\bibfnamefont {J.}~\bibnamefont {Sonner}},\ }\bibfield  {journal} {\bibinfo  {journal} {Physical Review Letters}\ }\textbf {\bibinfo {volume} {136}},\ \href {https://doi.org/10.1103/rvwb-r9lv} {10.1103/rvwb-r9lv} (\bibinfo {year} {2026}{\natexlab{a}})\BibitemShut {NoStop}%
\bibitem [{\citenamefont {Zhou}\ \emph {et~al.}(2026{\natexlab{b}})\citenamefont {Zhou}, \citenamefont {Zhou},\ and\ \citenamefont {Sonner}}]{Zhou2026}%
  \BibitemOpen
  \bibfield  {author} {\bibinfo {author} {\bibfnamefont {Y.-N.}\ \bibnamefont {Zhou}}, \bibinfo {author} {\bibfnamefont {T.-G.}\ \bibnamefont {Zhou}},\ and\ \bibinfo {author} {\bibfnamefont {J.}~\bibnamefont {Sonner}},\ }\href@noop {} {\bibinfo {title} {Three {H}amiltonians are sufficient for unitary $k$-design in temporal ensemble}} (\bibinfo {year} {2026}{\natexlab{b}}),\ \Eprint {https://arxiv.org/abs/2604.04205} {arXiv:2604.04205 [quant-ph]} \BibitemShut {NoStop}%
\bibitem [{\citenamefont {Sun}\ and\ \citenamefont {Zhang}(2026)}]{Sun2026}%
  \BibitemOpen
  \bibfield  {author} {\bibinfo {author} {\bibfnamefont {N.}~\bibnamefont {Sun}}\ and\ \bibinfo {author} {\bibfnamefont {P.}~\bibnamefont {Zhang}},\ }\href@noop {} {\bibinfo {title} {Unitary designs from two chaotic {H}amiltonians and a random {P}auli operation}} (\bibinfo {year} {2026}),\ \Eprint {https://arxiv.org/abs/2604.10122} {arXiv:2604.10122 [quant-ph]} \BibitemShut {NoStop}%
\bibitem [{\citenamefont {Cui}\ \emph {et~al.}(2025)\citenamefont {Cui}, \citenamefont {Schuster}, \citenamefont {Mao}, \citenamefont {Huang},\ and\ \citenamefont {Brand\~{a}o}}]{Cui2025}%
  \BibitemOpen
  \bibfield  {author} {\bibinfo {author} {\bibfnamefont {L.}~\bibnamefont {Cui}}, \bibinfo {author} {\bibfnamefont {T.}~\bibnamefont {Schuster}}, \bibinfo {author} {\bibfnamefont {L.}~\bibnamefont {Mao}}, \bibinfo {author} {\bibfnamefont {H.-Y.}\ \bibnamefont {Huang}},\ and\ \bibinfo {author} {\bibfnamefont {F.}~\bibnamefont {Brand\~{a}o}},\ }\href@noop {} {\bibinfo {title} {Random unitaries from hamiltonian dynamics}} (\bibinfo {year} {2025}),\ \Eprint {https://arxiv.org/abs/2510.08434} {arXiv:2510.08434 [quant-ph]} \BibitemShut {NoStop}%
\bibitem [{\citenamefont {Riddell}\ \emph {et~al.}(2023)\citenamefont {Riddell}, \citenamefont {Pagliaroli},\ and\ \citenamefont {Álvaro M.~Alhambra}}]{10.21468/SciPostPhys.15.4.165}%
  \BibitemOpen
  \bibfield  {author} {\bibinfo {author} {\bibfnamefont {J.}~\bibnamefont {Riddell}}, \bibinfo {author} {\bibfnamefont {N.~J.}\ \bibnamefont {Pagliaroli}},\ and\ \bibinfo {author} {\bibnamefont {Álvaro M.~Alhambra}},\ }\href {https://doi.org/10.21468/SciPostPhys.15.4.165} {\bibfield  {journal} {\bibinfo  {journal} {SciPost Phys.}\ }\textbf {\bibinfo {volume} {15}},\ \bibinfo {pages} {165} (\bibinfo {year} {2023})}\BibitemShut {NoStop}%
\bibitem [{\citenamefont {Huang}\ and\ \citenamefont {Harrow}(2020)}]{huang2020instabilitylocalizationtranslationinvariantsystems}%
  \BibitemOpen
  \bibfield  {author} {\bibinfo {author} {\bibfnamefont {Y.}~\bibnamefont {Huang}}\ and\ \bibinfo {author} {\bibfnamefont {A.~W.}\ \bibnamefont {Harrow}},\ }\href {https://arxiv.org/abs/1907.13392} {\bibinfo {title} {Instability of localization in translation-invariant systems}} (\bibinfo {year} {2020}),\ \Eprint {https://arxiv.org/abs/1907.13392} {arXiv:1907.13392 [cond-mat.dis-nn]} \BibitemShut {NoStop}%
\bibitem [{\citenamefont {Nakata}\ \emph {et~al.}(2014)\citenamefont {Nakata}, \citenamefont {Koashi},\ and\ \citenamefont {Murao}}]{Nakata_2014}%
  \BibitemOpen
  \bibfield  {author} {\bibinfo {author} {\bibfnamefont {Y.}~\bibnamefont {Nakata}}, \bibinfo {author} {\bibfnamefont {M.}~\bibnamefont {Koashi}},\ and\ \bibinfo {author} {\bibfnamefont {M.}~\bibnamefont {Murao}},\ }\href {https://doi.org/10.1088/1367-2630/16/5/053043} {\bibfield  {journal} {\bibinfo  {journal} {New Journal of Physics}\ }\textbf {\bibinfo {volume} {16}},\ \bibinfo {pages} {053043} (\bibinfo {year} {2014})}\BibitemShut {NoStop}%
\bibitem [{SM()}]{SM}%
  \BibitemOpen
  \href@noop {} {}\bibinfo {note} {See Supplemental Material for proofs of the theorems quoted in the main text, Porter-Thomas analysis for the product basis initial ensembles, derivation of the finite-$T$ correction, and stabilizer state initial ensemble.}\BibitemShut {Stop}%
\bibitem [{PT()}]{PT}%
  \BibitemOpen
  \href@noop {} {}\bibinfo {note} {In fact, we only require that the moments satisfy $\mathbb{E}[x_{nm}]=1/D$, $\mathbb{E}[x_{nm}^p]= \Theta(D^{-p})$, and that different $x_{nm}$ can be treated as independent variables. See SM for more details.}\BibitemShut {Stop}%
\bibitem [{\citenamefont {Srednicki}(1994)}]{Srednicki1994}%
  \BibitemOpen
  \bibfield  {author} {\bibinfo {author} {\bibfnamefont {M.}~\bibnamefont {Srednicki}},\ }\href {https://doi.org/10.1103/PhysRevE.50.888} {\bibfield  {journal} {\bibinfo  {journal} {Phys. Rev. E}\ }\textbf {\bibinfo {volume} {50}},\ \bibinfo {pages} {888} (\bibinfo {year} {1994})}\BibitemShut {NoStop}%
\bibitem [{\citenamefont {Rigol}\ \emph {et~al.}(2008)\citenamefont {Rigol}, \citenamefont {Dunjko},\ and\ \citenamefont {Olshanii}}]{Rigol2008}%
  \BibitemOpen
  \bibfield  {author} {\bibinfo {author} {\bibfnamefont {M.}~\bibnamefont {Rigol}}, \bibinfo {author} {\bibfnamefont {V.}~\bibnamefont {Dunjko}},\ and\ \bibinfo {author} {\bibfnamefont {M.}~\bibnamefont {Olshanii}},\ }\href {https://doi.org/10.1038/nature06838} {\bibfield  {journal} {\bibinfo  {journal} {Nature}\ }\textbf {\bibinfo {volume} {452}},\ \bibinfo {pages} {854} (\bibinfo {year} {2008})}\BibitemShut {NoStop}%
\bibitem [{\citenamefont {D'Alessio}\ \emph {et~al.}(2016)\citenamefont {D'Alessio}, \citenamefont {Kafri}, \citenamefont {Polkovnikov},\ and\ \citenamefont {Rigol}}]{DAlessio2016}%
  \BibitemOpen
  \bibfield  {author} {\bibinfo {author} {\bibfnamefont {L.}~\bibnamefont {D'Alessio}}, \bibinfo {author} {\bibfnamefont {Y.}~\bibnamefont {Kafri}}, \bibinfo {author} {\bibfnamefont {A.}~\bibnamefont {Polkovnikov}},\ and\ \bibinfo {author} {\bibfnamefont {M.}~\bibnamefont {Rigol}},\ }\href {https://doi.org/10.1080/00018732.2016.1198134} {\bibfield  {journal} {\bibinfo  {journal} {Adv. Phys.}\ }\textbf {\bibinfo {volume} {65}},\ \bibinfo {pages} {239} (\bibinfo {year} {2016})}\BibitemShut {NoStop}%
\bibitem [{\citenamefont {Atas}\ and\ \citenamefont {Bogomolny}(2017)}]{atas2015quantumisingmodeltransverse}%
  \BibitemOpen
  \bibfield  {author} {\bibinfo {author} {\bibfnamefont {Y.~Y.}\ \bibnamefont {Atas}}\ and\ \bibinfo {author} {\bibfnamefont {E.}~\bibnamefont {Bogomolny}},\ }\href {https://doi.org/10.1088/1751-8121/aa81f6} {\bibfield  {journal} {\bibinfo  {journal} {J. Phys. A: Math. Theor.}\ }\textbf {\bibinfo {volume} {50}},\ \bibinfo {pages} {385102} (\bibinfo {year} {2017})}\BibitemShut {NoStop}%
\bibitem [{\citenamefont {Mele}(2024)}]{Mele:2023ojv}%
  \BibitemOpen
  \bibfield  {author} {\bibinfo {author} {\bibfnamefont {A.~A.}\ \bibnamefont {Mele}},\ }\href {https://doi.org/10.22331/q-2024-05-08-1340} {\bibfield  {journal} {\bibinfo  {journal} {Quantum}\ }\textbf {\bibinfo {volume} {8}},\ \bibinfo {pages} {1340} (\bibinfo {year} {2024})},\ \Eprint {https://arxiv.org/abs/2307.08956} {arXiv:2307.08956 [quant-ph]} \BibitemShut {NoStop}%
\bibitem [{\citenamefont {Kim}\ and\ \citenamefont {Huse}(2013)}]{PhysRevLett.111.127205}%
  \BibitemOpen
  \bibfield  {author} {\bibinfo {author} {\bibfnamefont {H.}~\bibnamefont {Kim}}\ and\ \bibinfo {author} {\bibfnamefont {D.~A.}\ \bibnamefont {Huse}},\ }\href {https://doi.org/10.1103/PhysRevLett.111.127205} {\bibfield  {journal} {\bibinfo  {journal} {Phys. Rev. Lett.}\ }\textbf {\bibinfo {volume} {111}},\ \bibinfo {pages} {127205} (\bibinfo {year} {2013})}\BibitemShut {NoStop}%
\bibitem [{\citenamefont {Ba\~nuls}\ \emph {et~al.}(2011)\citenamefont {Ba\~nuls}, \citenamefont {Cirac},\ and\ \citenamefont {Hastings}}]{PhysRevLett.106.050405}%
  \BibitemOpen
  \bibfield  {author} {\bibinfo {author} {\bibfnamefont {M.~C.}\ \bibnamefont {Ba\~nuls}}, \bibinfo {author} {\bibfnamefont {J.~I.}\ \bibnamefont {Cirac}},\ and\ \bibinfo {author} {\bibfnamefont {M.~B.}\ \bibnamefont {Hastings}},\ }\href {https://doi.org/10.1103/PhysRevLett.106.050405} {\bibfield  {journal} {\bibinfo  {journal} {Phys. Rev. Lett.}\ }\textbf {\bibinfo {volume} {106}},\ \bibinfo {pages} {050405} (\bibinfo {year} {2011})}\BibitemShut {NoStop}%
\bibitem [{\citenamefont {Torres-Herrera}\ \emph {et~al.}(2015)\citenamefont {Torres-Herrera}, \citenamefont {Távora},\ and\ \citenamefont {Santos}}]{Torres_Herrera_2015}%
  \BibitemOpen
  \bibfield  {author} {\bibinfo {author} {\bibfnamefont {E.~J.}\ \bibnamefont {Torres-Herrera}}, \bibinfo {author} {\bibfnamefont {M.}~\bibnamefont {Távora}},\ and\ \bibinfo {author} {\bibfnamefont {L.~F.}\ \bibnamefont {Santos}},\ }\href {https://doi.org/10.1007/s13538-015-0366-3} {\bibfield  {journal} {\bibinfo  {journal} {Brazilian Journal of Physics}\ }\textbf {\bibinfo {volume} {46}},\ \bibinfo {pages} {239–247} (\bibinfo {year} {2015})}\BibitemShut {NoStop}%
\bibitem [{\citenamefont {Cotler}\ \emph {et~al.}(2017)\citenamefont {Cotler}, \citenamefont {Hunter-Jones}, \citenamefont {Liu},\ and\ \citenamefont {Yoshida}}]{Cotler_2017}%
  \BibitemOpen
  \bibfield  {author} {\bibinfo {author} {\bibfnamefont {J.}~\bibnamefont {Cotler}}, \bibinfo {author} {\bibfnamefont {N.}~\bibnamefont {Hunter-Jones}}, \bibinfo {author} {\bibfnamefont {J.}~\bibnamefont {Liu}},\ and\ \bibinfo {author} {\bibfnamefont {B.}~\bibnamefont {Yoshida}},\ }\bibfield  {journal} {\bibinfo  {journal} {Journal of High Energy Physics}\ }\textbf {\bibinfo {volume} {2017}},\ \href {https://doi.org/10.1007/jhep11(2017)048} {10.1007/jhep11(2017)048} (\bibinfo {year} {2017})\BibitemShut {NoStop}%
\bibitem [{\citenamefont {Fyodorov}(2010)}]{fyodorov2010introductionrandommatrixtheory}%
  \BibitemOpen
  \bibfield  {author} {\bibinfo {author} {\bibfnamefont {Y.~V.}\ \bibnamefont {Fyodorov}},\ }\href {https://arxiv.org/abs/math-ph/0412017} {\bibinfo {title} {Introduction to the random matrix theory: Gaussian unitary ensemble and beyond}} (\bibinfo {year} {2010}),\ \Eprint {https://arxiv.org/abs/math-ph/0412017} {arXiv:math-ph/0412017 [math-ph]} \BibitemShut {NoStop}%
\bibitem [{\citenamefont {Zhang}\ \emph {et~al.}(2026)\citenamefont {Zhang}, \citenamefont {Vijay}, \citenamefont {Gu},\ and\ \citenamefont {Bao}}]{Zhang_2026}%
  \BibitemOpen
  \bibfield  {author} {\bibinfo {author} {\bibfnamefont {Y.}~\bibnamefont {Zhang}}, \bibinfo {author} {\bibfnamefont {S.}~\bibnamefont {Vijay}}, \bibinfo {author} {\bibfnamefont {Y.}~\bibnamefont {Gu}},\ and\ \bibinfo {author} {\bibfnamefont {Y.}~\bibnamefont {Bao}},\ }\bibfield  {journal} {\bibinfo  {journal} {PRX Quantum}\ }\textbf {\bibinfo {volume} {7}},\ \href {https://doi.org/10.1103/myrb-nyhf} {10.1103/myrb-nyhf} (\bibinfo {year} {2026})\BibitemShut {NoStop}%
\bibitem [{\citenamefont {Gross}\ \emph {et~al.}(2021)\citenamefont {Gross}, \citenamefont {Nezami},\ and\ \citenamefont {Walter}}]{Gross_2021}%
  \BibitemOpen
  \bibfield  {author} {\bibinfo {author} {\bibfnamefont {D.}~\bibnamefont {Gross}}, \bibinfo {author} {\bibfnamefont {S.}~\bibnamefont {Nezami}},\ and\ \bibinfo {author} {\bibfnamefont {M.}~\bibnamefont {Walter}},\ }\href {https://doi.org/10.1007/s00220-021-04118-7} {\bibfield  {journal} {\bibinfo  {journal} {Communications in Mathematical Physics}\ }\textbf {\bibinfo {volume} {385}},\ \bibinfo {pages} {1325–1393} (\bibinfo {year} {2021})}\BibitemShut {NoStop}%
\bibitem [{\citenamefont {Bittel}\ \emph {et~al.}(2025)\citenamefont {Bittel}, \citenamefont {Eisert}, \citenamefont {Leone}, \citenamefont {Mele},\ and\ \citenamefont {Oliviero}}]{bittel2025completetheorycliffordcommutant}%
  \BibitemOpen
  \bibfield  {author} {\bibinfo {author} {\bibfnamefont {L.}~\bibnamefont {Bittel}}, \bibinfo {author} {\bibfnamefont {J.}~\bibnamefont {Eisert}}, \bibinfo {author} {\bibfnamefont {L.}~\bibnamefont {Leone}}, \bibinfo {author} {\bibfnamefont {A.~A.}\ \bibnamefont {Mele}},\ and\ \bibinfo {author} {\bibfnamefont {S.~F.~E.}\ \bibnamefont {Oliviero}},\ }\href {https://arxiv.org/abs/2504.12263} {\bibinfo {title} {A complete theory of the clifford commutant}} (\bibinfo {year} {2025}),\ \Eprint {https://arxiv.org/abs/2504.12263} {arXiv:2504.12263 [quant-ph]} \BibitemShut {NoStop}%
\bibitem [{\citenamefont {O'Donnell}(2021)}]{odonnell2021analysisbooleanfunctions}%
  \BibitemOpen
  \bibfield  {author} {\bibinfo {author} {\bibfnamefont {R.}~\bibnamefont {O'Donnell}},\ }\href {https://arxiv.org/abs/2105.10386} {\bibinfo {title} {Analysis of boolean functions}} (\bibinfo {year} {2021}),\ \Eprint {https://arxiv.org/abs/2105.10386} {arXiv:2105.10386 [cs.DM]} \BibitemShut {NoStop}%
\end{thebibliography}%

\onecolumngrid
\begin{center}
\bfseries End Matter
\end{center}
\vspace{0.6em}

\twocolumngrid

\setcounter{equation}{0}
\renewcommand{\theequation}{A\arabic{equation}}
\renewcommand{\theHequation}{A\arabic{equation}}

\subsection*{Frame potential and $k$-designs}
Here we summarize the relation between the frame potential and the additive error of approximate designs. For a state ensemble $\mathcal{E}$, its $k$-th moment operator is defined as
\begin{equation}
\rho_{\mathcal{E}}^{(k)}:=\mathop{\mathbb{E}}\limits_{|\psi\rangle\sim\mathcal{E}}\left[\left(|\psi\rangle\langle\psi|\right)^{\otimes k}\right].
\end{equation}
$\mathcal{E}$ forms a state $k$-design if and only if $\rho_{\mathcal{E}}^{(k)}=\rho_{\rm Haar}^{(k)}$, where $\rho_{\rm Haar}^{(k)}$ is the $k$-th moment operator of the Haar ensemble. $\mathcal{E}$ forms an approximate state $k$-design with additive error $\varepsilon$ if $\lVert\rho_{\mathcal{E}}^{(k)}-\rho_{\rm Haar}^{(k)}\rVert_1\leq \varepsilon$. The deviation of the state frame potential from its Haar value bounds the additive error as
\begin{equation}
\sqrt{F_{\mathcal{E}}^{(k)}-F_{\rm Haar}^{(k)}}\leq\lVert\rho_{\mathcal{E}}^{(k)}-\rho_{\rm Haar}^{(k)}\rVert_1\leq\sqrt{\frac{F_{\mathcal{E}}^{(k)}-F_{\rm Haar}^{(k)}}{F_{\rm Haar}^{(k)}}}.
\end{equation}
In particular, $\mathcal{E}$ is an exact state $k$-design if and only if $F_{\mathcal{E}}^{(k)}=F_{\rm Haar}^{(k)}$.

Similarly, for a unitary ensemble $\mathcal{E}$, its $k$-th moment channel is defined as
\begin{equation}
\Phi_{\mathcal{E}}^{(k)}(\cdot):=\mathop{\mathbb{E}}\limits_{U\sim\mathcal{E}}\left[U^{\otimes k}(\cdot)U^{\dagger\otimes k}\right].
\end{equation}
$\mathcal{E}$ forms a  unitary $k$-design if and only if $\Phi_{\mathcal{E}}^{(k)}=\Phi_{\rm Haar}^{(k)}$, and it forms an approximate unitary $k$-design with additive error $\varepsilon$ if $\lVert\Phi_{\mathcal{E}}^{(k)}-\Phi_{\rm Haar}^{(k)}\rVert_{\diamond}\leq \varepsilon$. The unitary frame potential bounds the additive error through~\cite{Mele:2023ojv}
\begin{equation}
\frac{\sqrt{\mathcal{F}_{\mathcal{E}}^{(k)}-\mathcal{F}_{\rm Haar}^{(k)}}}{D^{3k/2}}\leq\lVert\Phi_{\mathcal{E}}^{(k)}-\Phi_{\rm Haar}^{(k)}\rVert_{\diamond}\leq D^k\sqrt{\mathcal{F}_{\mathcal{E}}^{(k)}-\mathcal{F}_{\rm Haar}^{(k)}}.
\end{equation}
In particular, $\mathcal{E}$ is a unitary $k$-design if and only if $\mathcal{F}_{\mathcal{E}}^{(k)}=\mathcal{F}_{\rm Haar}^{(k)}$.

\setcounter{equation}{0}
\renewcommand{\theequation}{B\arabic{equation}}
\renewcommand{\theHequation}{B\arabic{equation}}

\subsection*{Numerical methods}

This section provides the numerical details for the results shown in Figs.~\ref{fig:N_scaling} and~\ref{fig:T_scaling}. For Fig.~\ref{fig:N_scaling}, we consider the ensemble $\mathcal{E}=\{e^{-iHt}|y\rangle\mid t\sim \mathrm{Unif}[0,T], |y\rangle\in Y\text{-basis}\}$, where $H$ is the mixed-field Ising Hamiltonian~\eqref{eq:MFIM} with parameters $(J,h_x,h_z)=(1,0.9045,0.809)$~\cite{PhysRevLett.111.127205}. The corresponding frame potential is $F_{\mathcal{E}}^{(k)}=\mathbb{E}_{t,t',y,y'}|\langle y|e^{iHt}e^{-iHt'}|y'\rangle|^{2k}$.

We evaluate this quantity by exact diagonalization. To improve the sampling efficiency, we use a stratified sampling scheme according to whether $y=y'$ and whether the time difference $|t-t'|$ is smaller than a cutoff $t_0$. Explicitly,
\begin{equation}
\begin{aligned}
&F_{\mathcal{E}}^{(k)}=\frac{1}{D}p(t_0)\mathop{\mathbb{E}}\limits_{y,|t-t'|\leq t_0}\left|\langle y|e^{iHt}e^{-iHt'}|y\rangle\right|^{2k}\\
&+\frac{1}{D}\left[1-p(t_0)\right]\mathop{\mathbb{E}}\limits_{y, |t-t'|>t_0}\left|\langle y|e^{iHt}e^{-iHt'}|y\rangle\right|^{2k}\\
&+\left(1-\frac{1}{D}\right)p(t_0)\mathop{\mathbb{E}}\limits_{y\neq y', |t-t'|\leq t_0}\left|\langle y|e^{iHt}e^{-iHt'}|y'\rangle\right|^{2k}\\
&+\left(1-\frac{1}{D}\right)\left[1-p(t_0)\right]\mathop{\mathbb{E}}\limits_{y\neq y', |t-t'|>t_0}\left|\langle y|e^{iHt}e^{-iHt'}|y'\rangle\right|^{2k}.
\end{aligned}
\end{equation}
Here $p(t_0)=(2Tt_0-t_0^2)/T^2$ is the probability that two independently sampled times $t,t'\in[0,T]$ satisfy $|t-t'|\leq t_0$. In Fig.~\ref{fig:N_scaling}, we take $T=10^6$ and $t_0=2$. To estimate the long-time limit frame potential shown in Fig.~\ref{fig:N_scaling}, we keep only the contributions with $|t-t'|>t_0$. For each of the two large-time-difference strata, $y=y'$ and $y\neq y'$, we use $5\times 10^7$ samples.

The stratification is useful because the overlap $|\bra{y}e^{iH(t-t')}\ket{y'}|$ is large when $|t-t'|$ is small. In particular, when $y=y'$, the overlap is close to unity for $|t-t'|\ll 1$. In addition, for certain pairs $y\neq y'$, the transition amplitude $\bra{y}H\ket{y'}$ can also develop peaks at small but nonzero time differences. Although these short-time events occur with small probability, they can give significant contributions after taking the $2k$-th power and are therefore easily underestimated by naive uniform sampling.

For Fig.~\ref{fig:T_scaling}, we set $N=10$. We use the same parameters for $H_1$ as in Fig.~\ref{fig:N_scaling} and choose $H_2$ with parameters $(J,h_x,h_z)=(1,-1.05,0.5)$~\cite{PhysRevLett.106.050405}. We use the same stratified sampling strategy as before. For example, in the two-step case, we first divide the samples into two strata, $y=y'$ and $y\neq y'$. We then further split each stratum according to whether each of $|t_1-t_1'|$ and $|t_2-t_2'|$ is smaller or larger than $t_0$, resulting in eight strata in total. Combining these strata with their corresponding probabilities gives the finite-$T$ frame potential. We again take $t_0=2$.

\clearpage
\newpage
\onecolumngrid
\appendix 

\section*{Supplemental Material}

\setcounter{section}{0}
\renewcommand{\thesection}{\arabic{section}}

\counterwithout{equation}{section}
\setcounter{equation}{0}
\renewcommand{\theequation}{S\arabic{equation}}
\renewcommand{\theHequation}{S\arabic{equation}}

\setcounter{figure}{0}
\renewcommand{\thefigure}{S\arabic{figure}}
\renewcommand{\theHfigure}{S\arabic{figure}}

\setcounter{theorem}{0}
\renewcommand{\thetheorem}{S\arabic{theorem}}
\renewcommand{\theHtheorem}{S\arabic{theorem}}

The Supplemental Material is organized as follows. In Appendix~\ref{ap:Theorem1}, we derive Eq.~(2) in the main text and prove Theorem~1. In Appendix~\ref{ap:PT}, we use Porter-Thomas statistics to show the evolved ensemble $\mathcal{E}$ can form a state $k$-design when the initial ensemble is a product-state basis. In Appendix~\ref{ap:Theorem2}, we first review the no-go theorem derived in Ref.~\cite{Cui2025}, and then prove our no-go theorem, Theorem~2 in the main text. In Appendix~\ref{ap:ldos}, we show the energy population of several Pauli-basis states for the mixed-field Ising Hamiltonian, providing a physical explanation for choosing the $Y$-basis as the initial ensemble. In Appendix~\ref{ap:finiteT}, we derive the finite-T correction to the GUE-averaged frame potential. In Appendix~\ref{ap:stabilizer}, we consider the initial ensemble to be the random stabilizer-state ensemble and prove that diagonal Hamiltonians can evolve it into a state $k$-design. In Appendix~\ref{ap:Theorem3}, we prove Theorem~3 in the main text.

\section{Derivation of Eq.~(2) and proof of Theorem~1}
\label{ap:Theorem1}

In this section, we first derive Eq.~(2) of the main text and then prove Theorem~1.

We begin with the ensemble $\mathcal{E}=\{e^{-iHt}\ket{\psi}\mid t\sim \mathrm{Unif}[0,T],\ket{\psi}\sim \mathcal{E}'\}$, whose frame potential is
\begin{equation}
    F^{(k)}_{\mathcal{E}}=\mathop{\mathbb{E}}\limits_{\ket{\psi} ,\ket{\phi} \sim \mathcal{E}'} \int _{0}^{T}\frac{dtdt'}{T^{2}} |\bra{\psi}e^{iHt}e^{-iHt'}\ket{\phi}|^{2k}.
\end{equation}
Expanding the Hamiltonian $H$ in the eigenbasis $H=\sum _{n}E_{n}\ket{E_{n}}\bra{E_{n}}$, we obtain
\begin{equation}
F^{(k)}_{\mathcal{E}}=\mathop{\mathbb{E}}\limits_{\ket{\psi} ,\ket{\phi} \sim \mathcal{E}'}\sum _{\vec{n},\vec{m}}\int _{0}^{T}\frac{dtdt'}{T^{2}}\exp\left[ i(t-t')\left(\sum _{a=1}^{k}E_{n_{a}}-\sum _{a=1}^{k}E_{m_{a}}\right) \right]\prod _{a=1}^{k}\braket{ E_{m_{a}} }{ \psi }\braket{ \psi }{ E_{n_{a}} } \braket{ E_{n_{a}} }{ \phi } \braket{ \phi }{ E_{m_{a}} } .
\end{equation}
The time integral part equals
\begin{equation}
    \int _{0}^{T}\frac{dtdt'}{T^{2}}\exp\left[ i(t-t') \Delta E\right]=\left(\frac{\sin (\Delta ET/2)}{\Delta ET/2}\right)^2,
\end{equation}
where $\Delta E:= \sum _{a=1}^{k}E_{n_{a}}-\sum _{a=1}^{k}E_{m_{a}}$. In the limit $T\to \infty$, it becomes $\delta _{\Delta E,0}$. Assuming the $k$-th no-resonance condition, $\Delta E=0$ if and only if there exists a permutation $\sigma \in S_{k}$ such that $m_{a}=n_{\sigma(a)}$ for $a=1,2,\cdots,k$. Thus the sums over $\vec{n}$ and $\vec{m}$ can be replaced by sums over $\vec{n}$ and the permutations.
\begin{equation}
\begin{aligned}
F^{(k)}_{\mathcal{E}}&=\mathop{\mathbb{E}}\limits_{\ket{\psi} ,\ket{\phi} \sim \mathcal{E}'}\sum _{\vec{n}}\sum _{\sigma \in S_{k}}\frac{1}{\Omega(\vec{n})}\prod _{a=1}^{k}\braket{ E_{n_{\sigma(a)}} }{ \psi }\braket{ \psi }{ E_{n_{a}} } \braket{ E_{n_{a}} }{ \phi } \braket{ \phi }{ E_{n_{\sigma(a)}} }\\
&=\mathop{\mathbb{E}}\limits_{\ket{\psi} ,\ket{\phi} \sim \mathcal{E}'}\sum _{\vec{n}}\frac{k!}{\Omega(\vec{n})}\prod _{a=1}^{k} \braket{ E_{n_{a}} }{ \psi }\braket{ \psi }{ E_{n_{a}} }\braket{ E_{n_{a}} }{ \phi } \braket{ \phi }{ E_{n_{a}} }\\
&=k!\sum _{\vec{n}}\frac{1}{\Omega(\vec{n})}\left( \bra{E_{n_{1}}\cdots E_{n_{k}}} \rho ^{(k)}_{\mathcal{E}'}\ket{E_{n_{1}}\cdots E_{n_{k}}}  \right) ^{2}.
\label{eq:sm_derive_eq2}
\end{aligned}
\end{equation}
This is precisely Eq.~(2) in the main text. The factor \(1/\Omega(\vec n)\) corrects the overcounting caused by repeated indices in \(\vec n\). Indeed, if some indices \(n_a\) coincide, different permutations \(\sigma\in S_k\) can generate the same tuple \((n_{\sigma(1)},\ldots,n_{\sigma(k)})\). For a tuple \(\vec n\) with multiplicities \(r_i=\sum_{a=1}^k\delta_{i,n_a}\), each distinct tuple is counted \(\Omega(\vec n)=\prod_i r_i!\) times. Although this factor is necessary for the exact expression, it does not affect the leading contribution in $D$ in the GUE-averaged calculation below, which comes from tuples with pairwise distinct indices \(n_1,\ldots,n_k\).

Now we prove Theorem~1 in the main text. We will use the following standard Weingarten formula. For any operator \(O\) acting on \((\mathbb C^D)^{\otimes q}\),
\begin{equation}
    \mathop{\mathbb{E}}\limits_{U\sim\mathrm{Haar}}
    \left[
        U^{\otimes q} O U^{\dagger\otimes q}
    \right]
    =
    \sum_{\pi,\sigma\in S_q}
    \mathrm{Wg}(\pi^{-1}\sigma,D)
    \operatorname{Tr}\!\left[V_D^\dagger(\sigma)O\right]
    V_D(\pi).
\end{equation}
Here \(V_D(\pi)\) denotes the permutation operator associated with \(\pi\in S_q\). In the large-\(D\) limit with fixed \(q\), the Weingarten function satisfies
\begin{equation}
    \mathrm{Wg}(\mathrm{id},D)
    =
    D^{-q}+O(D^{-q-2}),
    \qquad
    \mathrm{Wg}(\tau,D)
    =
    O\!\left(D^{-2q+\#\tau}\right)
    \quad
    (\tau\neq \mathrm{id}),
\end{equation}
where \(\#\tau\) denotes the number of cycles of the permutation \(\tau\). 

Since Eq.~\eqref{eq:sm_derive_eq2} depends on the Hamiltonian only through its eigenbasis, the GUE average reduces to a Haar average over the eigenvectors. Writing \(|E_n\rangle=U|n\rangle\) with \(U\sim \mathrm{Haar}\), and applying the Weingarten formula, we obtain
\begin{equation}
    \begin{aligned}
        \mathop{\mathbb{E}}\limits_{H\sim \mathrm{GUE}}F^{(k)}_{\mathcal{E}}&=\mathop{\mathbb{E}}\limits_{\ket{\psi} ,\ket{\phi} \sim \mathcal{E}'}k! \sum_{\vec{n}} \frac{1}{\Omega(\vec{n})}\bra{n_{1} \dots n_{k} n_{1} \dots n_{k}} \mathop{\mathbb{E}}\limits_{U\sim\mathrm{Haar}}[U^{\otimes 2k} (\ket{\psi}\bra{\psi}^{\otimes k} \otimes \ket{\phi}\bra{\phi}^{\otimes k}) U^{\dagger \otimes 2k}] \ket{n_{1} \dots n_{k} n_{1} \dots n_{k}}\\
        &=k!\sum_{\pi,\sigma \in S_{2k}}\mathrm{Wg}(\pi^{-1}\sigma,D)\sum_{\vec{n}}\frac{1}{\Omega(\vec{n})}\bra{n_{1} \dots n_{k} n_{1} \dots n_{k}}V_D(\pi)\ket{n_{1} \dots n_{k} n_{1} \dots n_{k}}\\
&\hspace{5cm}\times\mathop{\mathbb{E}}\limits_{\ket{\psi} ,\ket{\phi} \sim \mathcal{E}'}(\bra{\psi}^{\otimes k}\otimes\bra{\phi}^{\otimes k})V_D(\sigma^{-1})(\ket{\psi}^{\otimes k}\otimes\ket{\phi}^{\otimes k}).
\label{eq:sm_stateGUE}
    \end{aligned}
\end{equation}
Now we examine the leading order of Eq.~\eqref{eq:sm_stateGUE}. For fixed \(\pi\) and \(\sigma\), the contribution is a product of three terms. We first consider the second term, which is maximized when \(\pi\) only contains swaps between the \(a\)-th element and the \((a+k)\)-th element \((a=1,2,\ldots,k)\). In this case,
\begin{equation}
    \sum_{\vec{n}}\frac{1}{\Omega(\vec{n})}\bra{n_{1} \dots n_{k} n_{1} \dots n_{k}}V_D(\pi)\ket{n_{1} \dots n_{k} n_{1} \dots n_{k}}=\sum_{\vec{n}}\frac{1}{\Omega(\vec{n})}=D^k+O(D^{k-1}).
\end{equation}
For any other \(\pi\), there will be at least one constraint of the form \(n_i=n_j\), and hence the second term is at most \(O(D^{k-1})\). The permutations that only contain swaps between the \(a\)-th element and the \((a+k)\)-th element form a subgroup  of $S_{2k}$, which we denote by \(E_{2k}\). The first term is the Weingarten function, which is maximized when \(\pi=\sigma\), giving \(\mathrm{Wg}(\pi^{-1}\sigma,D)=D^{-2k}+O(D^{-2k-1})\). The third term is the frame potential of \(\mathcal E'\) of some order. Adding all terms with \(\pi=\sigma\in E_{2k}\) and keeping only the leading order, we obtain
\begin{equation}
    \frac{k!}{D^k}\left(1+\sum_{m=1}^k \binom{k}{m}F^{(m)}_{\mathcal{E}'}\right).
\end{equation}
Below we show that this is precisely the leading contribution of Eq.~\eqref{eq:sm_stateGUE}. The third term always gives the frame potential of some order. Therefore, it is enough to show that for all other choices of \(\pi\) and \(\sigma\), the product of the first two terms is  at most \(O(D^{-k-1})\). There are two cases. If \(\pi\notin E_{2k}\), the second term is at most \(O(D^{k-1})\), while the Weingarten function is at most \(O(D^{-2k})\). Their product is therefore at most \(O(D^{-k-1})\). If \(\pi\in E_{2k}\) but \(\sigma\neq\pi\), the second term is \(O(D^k)\), but the Weingarten function is at most $O(D^{-2k-1})$. Their product is again at most \(O(D^{-k-1})\). Hence all terms except those with \(\pi=\sigma\in E_{2k}\) are subleading by at least one power of \(D\). This proves Theorem~1 in the main text.

\section{Porter-Thomas analysis for the product-basis initial ensemble}
\label{ap:PT}
In this section, we use the Porter-Thomas distribution to show that the frame potential of the evolved ensemble $\mathcal{E}$ with a product-basis initial ensemble $\mathcal{E}'$ matches the Haar value in the leading order. We start with Eq.~(6) of the main text,
\begin{equation}
    F^{(k)}_{\mathcal{E}}=\frac{k!}{D^{2}}\sum _{m,m'}\sum _{\vec{n}}\frac{1}{\Omega(\vec{n})}\prod _{a=1}^{k} |  \langle E_{n_{a}}|m\rangle |^{2}|\langle E_{n_{a}}|m'\rangle|^{2}.
    \label{eq:sm_fpbasis}
\end{equation}
We assume that the eigenstate overlaps follow Porter-Thomas statistics. More precisely, we treat \(x_{nm}:=|\langle E_n|m\rangle|^2\), with \(n,m=1,\ldots,D\), as independent random variables. For complex eigenstate wavefunctions, \(x_{nm}\) has probability density \(p(x)=D\exp(-Dx)\), while for real eigenstate wavefunctions, \(p(x)=\sqrt{D/(2\pi x)}\exp(-Dx/2)\). Correspondingly, \(\mathbb E[x_{nm}^p]=p!/D^p\) in the complex case and \(\mathbb E[x_{nm}^p]=(2p-1)!!/D^p\) in the real case. In the following, we mainly use the properties \(\mathbb E[x_{nm}]=1/D\) and \(\mathbb E[x_{nm}^p]=\Theta(D^{-p})\), which hold in both cases. Therefore the two cases can be treated uniformly. 

Now we calculate the expectation value of $F^{(k)}_{\mathcal{E}}$. Dividing it into the $m=m'$ part and the $m\neq m'$ part, we obtain
\begin{equation}
    \mathbb{E}[F^{(k)}_{\mathcal{E}}]=\frac{k!}{D^{2}}\sum_m \sum _{\vec{n}}\frac{1}{\Omega(\vec{n})}\mathbb{E}\left[\prod_{a=1}^k x_{n_am}^2\right]+\frac{k!}{D^{2}}\sum_{m\neq m'}\sum _{\vec{n}}\frac{1}{\Omega(\vec{n})}\mathbb{E}\left[\prod_{a=1}^k x_{n_am}x_{n_am'}\right].
    \label{eq:sm_ptaverage}
\end{equation}
We first estimate the $m=m'$ part in Eq.~\eqref{eq:sm_ptaverage}. Since \(\mathbb E[x_{nm}^p]=\Theta(D^{-p})\), \(\mathbb E\!\left[\prod_{a=1}^k x_{n_a m}^2\right]=\Theta(D^{-2k})\) holds for any choice of \(n_1,\ldots,n_k\), including cases with repeated indices. Thus the leading order of this term comes from $n_1,\cdots,n_k$ being pairwise distinct.
\begin{equation}
    \begin{aligned}
        \frac{k!}{D^{2}}\sum_m \sum _{\vec{n}}\frac{1}{\Omega(\vec{n})}\mathbb{E}\left[\prod_{a=1}^k x_{n_am}^2\right]&=\frac{k!}{D^{2}}\cdot D\cdot\left( D(D-1)\cdots(D-k+1)\frac{\xi^k}{D^{2k}}+O(1/D^{k+1})\right)\\
        &=\frac{k!\xi^k}{D^{k+1}}+O(1/D^{k+2}).
    \end{aligned}
\end{equation}
Here $\xi=2$ in the complex case and $\xi=3$ in the real case, coming from the second-order moment of $x_{nm}$. The $m\neq m'$ part in Eq.~\eqref{eq:sm_ptaverage} can be analyzed similarly. 
\begin{equation}
    \begin{aligned}
        \frac{k!}{D^{2}}\sum_{m\neq m'}\sum _{\vec{n}}\frac{1}{\Omega(\vec{n})}\mathbb{E}\left[\prod_{a=1}^k x_{n_am}x_{n_am'}\right]&=\frac{k!}{D^{2}}\cdot(D^2-D)\cdot \left(D(D-1)\cdots(D-k+1)\frac{1}{D^{2k}}+O(1/D^{k+1})\right)\\
        &=\frac{k!}{D^k}+O(1/D^{k+1}).
    \end{aligned}
\end{equation}
Combining these two parts, we find the leading order contribution to $\mathbb{E}[F^{(k)}_{\mathcal{E}}]$ comes from the $m\neq m'$ part and equals  $k!/D^k$, matching the Haar value at leading order. 

For comparison, if the initial ensemble $\mathcal{E}'$ consists of  a single state $\ket{\psi}$ and we make the same Porter-Thomas assumption for the overlaps $|\braket{E_n}{\psi}|^2$, the leading order of $F^{(k)}_{\mathcal{E}}$ is $k!\xi^k/D^k$ after performing the same calculation, which deviates from the Haar value.

\section{Proof of Theorem~2}
\label{ap:Theorem2}

In this section, we first review a no-go theorem introduced in Ref.~\cite{Cui2025}, and then prove  Theorem~2 in the main text.

Although Ref.~\cite{Cui2025} studies unitary ensembles, its proof directly implies a no-go theorem for certain state ensembles. We record this state-ensemble version as follows.
\begin{theorem}[Theorem~1 in Ref.~\cite{Cui2025}]
    Consider a product-state ensemble $\mathcal{E}':=\{\ket{u_1}\otimes\cdots\otimes\ket{u_N}\}$, where each $\ket{u_i}$ is sampled independently from an exact one-qubit state $2$-design. Consider the ensemble $\mathcal{E}:=\{e^{-iHt}\ket{\psi}\mid (H,t)\sim\mathcal{D},\ket{\psi}\sim \mathcal{E}'\}$, where $\mathcal{D}$ is an arbitrary distribution over Hamiltonians and evolution times, and all Hamiltonians in $\mathcal{D}$ are $q$-local. Then $\mathcal{E}$ cannot form an approximate state $2$-design with additive error $\varepsilon<\frac{1}{12^qN}$ for $1$-D geometry, or $\varepsilon<\frac{1}{12^qN^q}$ for all-to-all geometry.  
\end{theorem}
This theorem precludes choosing the initial ensemble to be, for example, a tensor product of one-qubit Haar random states or one-qubit stabilizer states,  if one requires an exponentially small additive error.  

We now prove Theorem~2 in the main text, which provides a more useful criterion for choosing an initial ensemble compatible with a given Hamiltonian. Consider the ensemble $\mathcal{E}=\{e^{-iHt}\ket{\psi}\mid t\sim\mathcal{D},\ket{\psi}\sim\mathcal{E}'\}$, where $\mathcal{D}$ is an arbitrary time distribution. Define a quantity
\begin{equation}
    E=\mathop{\mathbb{E}}_{\ket{\phi}\sim \mathcal{E}}\bra{\phi}H\ket{\phi}^2.
\end{equation}
On the one hand, $E=\mathrm{tr}[H^{\otimes 2}\rho ^{(2)}_{\mathcal{E}}]$. If $\mathcal{E}$ forms an approximate state 2-design with additive error $\varepsilon$, we obtain
\begin{equation}
        |E-\mathrm{tr}[H^{\otimes 2}\rho ^{(2)}_{\mathrm{Haar}}]|=|\mathrm{tr}[H^{\otimes 2}(\rho^{(2)}_{\mathcal{E}}-\rho ^{(2)}_{\mathrm{Haar}})]|\leq \varepsilon||H^{\otimes 2}||_{\infty}=\varepsilon||H||_\infty^2.
\end{equation}
On the other hand,
\begin{equation}
E=\mathop{\mathbb{E}}_{\ket{\psi}\sim\mathcal{E}',t\sim\mathcal{D}}\bra{\psi}e^{iHt}He^{-iHt}\ket{\psi}^2=\mathop{\mathbb{E}}_{\ket{\psi}\sim\mathcal{E}'}\bra{\psi}H\ket{\psi}^2.
\end{equation}
Thus we have
\begin{equation}
    \varepsilon\geq \frac{|\mathop{\mathbb{E}}_{\ket{\psi}\sim\mathcal{E}'}\bra{\psi}H\ket{\psi}^2-\mathrm{tr}[H^{\otimes 2}\rho ^{(2)}_{\mathrm{Haar}}]|}{||H||_\infty^2}.
\end{equation}
Expand $H$ in the Pauli basis as $H=\sum_{P\in\mathcal{P}_N}h_PP$. Using $\rho^{(2)}_{\mathrm{Haar}}=(\mathbb{I}+\mathrm{SWAP})/[D(D+1)]$ and $\mathrm{tr}(H)=0$, we get
\begin{equation}
    \mathrm{tr}[H^{\otimes 2}\rho ^{(2)}_{\mathrm{Haar}}]=\frac{\mathrm{tr}(H^2)}{D(D+1)}=\frac{\sum_{P\in\mathcal{P}_N}h_P^2}{D+1}.
\end{equation}
This proves the theorem.

\section{Energy distribution for bitstrings in the $X$-, $Y$- and $Z$-bases}
\label{ap:ldos}

\begin{figure}[t]
    \centering
    \includegraphics[width=\linewidth]{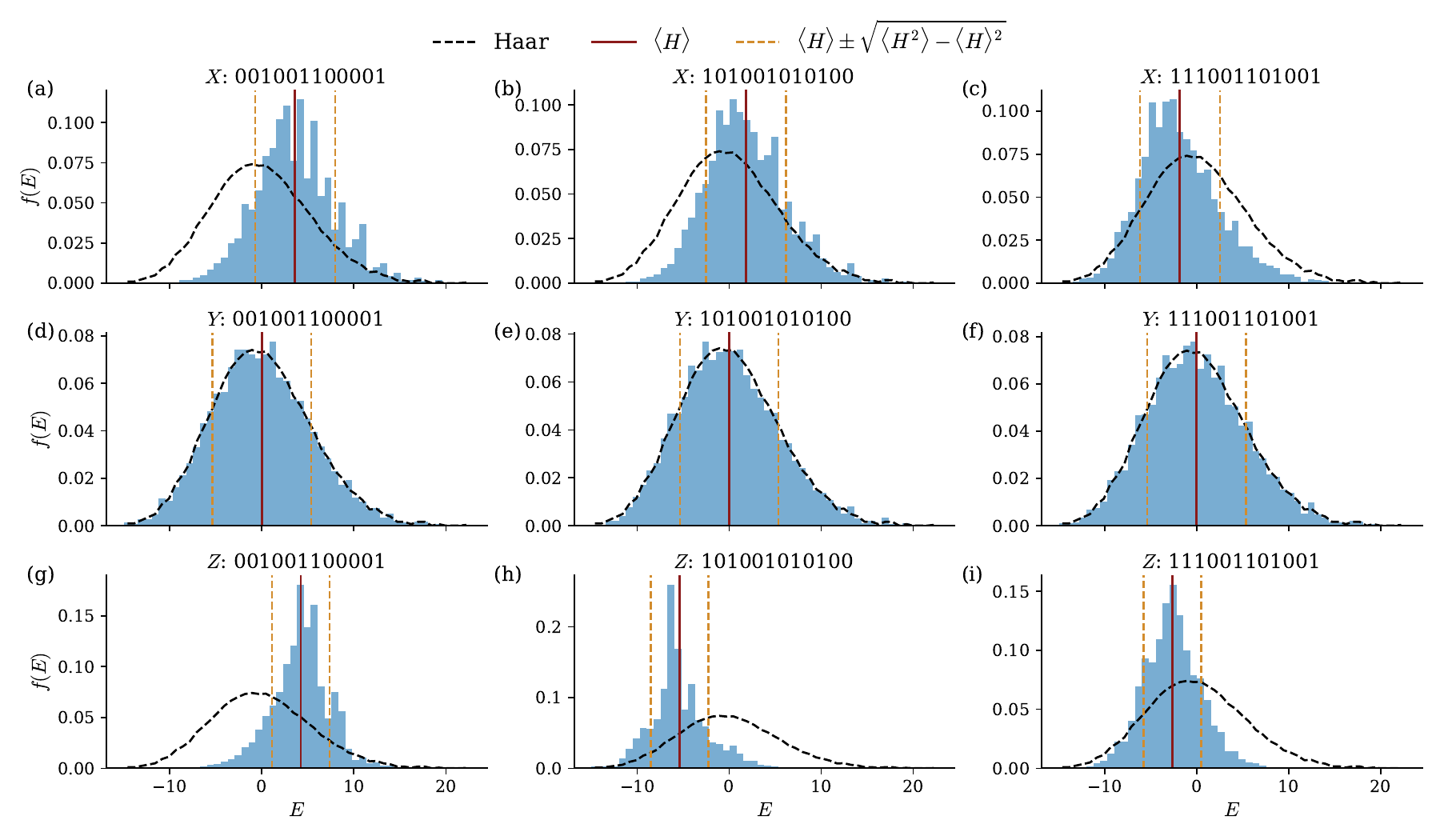}
    \caption{Local density of states $f_{{m}}(E)$ for several states chosen from the $X$-basis (a-c), $Y$-basis (d-f) and $Z$-basis (g-i) for the mixed-field Ising Hamiltonian with parameters $J=1$, $h_x=0.9045$, $h_z=0.809$ and $N=12$. The black dashed curve denotes the infinite-temperature (Haar-averaged) LDOS $f_{\mathrm{Haar}}(E)$. The red vertical line marks the mean energy $\bra{m}H\ket{m}$ and the orange vertical dashed line marks $\bra{m}H\ket{m}\pm \sqrt{\bra{m}H^2\ket{m}-\bra{m}H\ket{m}^2}$. }
    \label{fig:sm_energydistri}
\end{figure}

In this section, we provide an energy-distribution diagnostic that helps explain why the $Y$-basis is chosen as the  initial ensemble for generating $k$-designs under the mixed-field Ising Hamiltonian, while the $X$- and $Z$-bases are not. This diagnostic is complementary to the bound derived from energy conservation in the main text (Theorem~2).

For a state $\ket{m}$, we define a distribution
\begin{equation}
    f_{{m} }(E)=\sum _{n}|\braket{ m }{ E_{n} } |^{2}\delta(E-E_{n}).
\end{equation}
This distribution is often referred to as the local density of states (LDOS)~\cite{Torres_Herrera_2015}. The $k$-th moment of this distribution is
\begin{equation}
    \int E^{k}f_{{m} }(E)dE=\bra{m} H^{k}\ket{m} .
\end{equation}
The Haar-averaged LDOS is defined as
\begin{equation}
    f_{\mathrm{Ha a r}}(E)=\mathop{\mathbb{E}}_{\ket{m} \sim \mathrm{Ha ar}}f_{{m} }(E)=\sum _{n}\frac{1}{D}\delta(E-E_{n}),
\end{equation}
which is simply the normalized spectral density, corresponding to an infinite temperature ensemble. Its $k$-th moment is
\begin{equation}
    \int E^{k}f_{\mathrm{Ha ar} }(E)dE=\frac{1}{D}\mathrm{tr}(H^{k})=\mathop{\mathbb{E}}_{\ket{m} \sim \mathrm{Ha ar}}\bra{m} H^{k}\ket{m} .
\end{equation}

Figure~\ref{fig:sm_energydistri} shows the LDOS of several states $\ket{m}$ chosen from the $X$-, $Y$- and $Z$-bases  for the mixed-field Ising Hamiltonian.
We also show the infinite-temperature (Haar-averaged) distribution for comparison. The LDOS of the $Y$-basis states closely follows the infinite-temperature distribution, whereas clear deviations are visible for the $X$- and $Z$-basis states. This behavior can be understood from the low-order energy moments. For every $Y$-basis state $\ket{y}$, $\bra{y}H\ket{y}=0$ and $\bra{y}H^{2}\ket{y}=(1/D)\mathrm{tr}(H^{2})$, independent of the particular basis state. Therefore, each $Y$-basis state predominantly samples highly excited eigenstates in the bulk of the spectrum with the largest density of states, where chaotic behavior and random-matrix-like eigenstate statistics are expected. In contrast, the energy population for $Z$- and $X$-basis states varies strongly among the specific configurations. Consequently, the eigenstates they probe are strongly modulated by the average energy, and differ from the universal random matrix behavior.

\section{Derivation of the finite-$T$ correction}
\label{ap:finiteT}
In this section, we consider the ensemble $\mathcal{E}=\{e^{-iHt}\ket{\psi}\mid t\sim \mathrm{Unif}[0,T],\ket{\psi}\sim\mathcal{E}'\}$ and show the finite-$T$ correction to the GUE-averaged frame potential is $O(F^{(k)}_{\mathcal{E}'}/T)$.

First, we write the frame potential $F^{(k)}_{\mathcal{E}}$ in an integral form.
\begin{equation}
    F^{(k)}_{\mathcal{E}}=\frac{2}{T}\int _{0}^{T}d\tau\left( 1-\frac{\tau}{T} \right) f ^{(k)}_{\mathcal{E}'}(\tau)
\end{equation}
where $f ^{(k)}_{\mathcal{E}'}(\tau):=\mathbb{E}_{\ket{\psi},\ket{\phi}\sim\mathcal{E}'}|\bra{\psi}e^{iH\tau}\ket{\phi}|^{2k}$. Expanding $H$ in the eigenbasis, we get
\begin{equation}
    f ^{(k)}_{\mathcal{E}'}(\tau)=\sum_{\vec{n},\vec{n}'}\exp\left[i\tau\left(\sum_{a=1}^k E_{n_a}-\sum_{a=1}^kE_{n'_a}\right)\right]\mathop{\mathbb{E}}_{\ket{\psi},\ket{\phi}\sim\mathcal{E}'}\bra{E_{\vec{n}}E_{\vec{n}'}}\phi^k\psi^k\rangle\braket{\phi^k\psi^k}{E_{\vec{n}'}E_{\vec{n}}}
\end{equation}
where $\ket{E_{\vec{n}}E_{\vec{n}'}}:=\ket{E_{n_1}\cdots  E_{n_k}E_{n'_1}\cdots E_{n'_k}}$ and $\ket{\phi^k\psi^k}:=\ket{\phi}^{\otimes k }\otimes\ket{\psi}^{\otimes k}$. For the GUE ensemble, the eigenvalues and eigenvectors are independent. We therefore separate the GUE average into an eigenvalue part and an eigenvector part. Define
\begin{equation}
    g_{\vec{n},\vec{n}'}(\tau):=\int DE\exp\left[ i\left( \sum _{a=1}^kE_{n_{a}}-\sum _{a=1}^kE_{n'_{a}} \right)\tau \right]
\end{equation}
and
\begin{equation}
    G(\vec{n},\vec{n}',T):=\frac{2}{T}\int _{0}^{T}d\tau\left( 1-\frac{\tau}{T} \right)g_{\vec{n},\vec{n}'}(\tau),
\end{equation}
where $DE=dE_1\cdots dE_DP(E_1,\cdots,E_D)$ is the measure of the GUE eigenvalues, satisfying $\int DE=1$. Also define
\begin{equation}
    M(\vec{n},\vec{n}'):=\mathop{\mathbb{E}}_{\ket{\psi},\ket{\phi}\sim \mathcal{E}'}\bra{\vec{n},\vec{n}'} \mathop{\mathbb{E}}_{U\sim\mathrm{Haar}}[U^{\otimes 2k}\ket{\phi^k\psi^k}\bra{\phi^k\psi^k}U^{\dagger \otimes 2k}]\ket{\vec{n}',\vec{n}}.
\end{equation}
Using these definitions, the GUE-averaged frame potential can be written as
\begin{equation}
    \mathop{\mathbb{E}}_{H\sim \mathrm{GUE}}F^{(k)}_{\mathcal{E}}=\sum _{\vec{n},\vec{n}'}G(\vec{n},\vec{n}',T)M(\vec{n},\vec{n}')
    \label{eq:sm_GUEstatefp}
\end{equation}

For each pair $\vec{n}$ and $\vec{n}'$, we can define a partition of the set $\{1,2,\cdots,2k\}$. A partition of a set is a way of dividing a set into several disjoint subsets. These subsets are called blocks. For example, $P=\{\{1,k+1\},\{2,k+2\},\cdots,\{k,2k\}\}$ is a partition of the above set whose blocks are $\{1,k+1\},\{2,k+2\},\cdots,\{k,2k\}$. Given a pair of   $\vec{n}$ and $\vec{n}'$, define $\vec{x}:=(\vec{n},\vec{n}')$, which means $x_i=n_i$ for $1\leq i\leq k$ and $x_i=n'_{i-k}$ for $k+1\leq i\leq 2k$. $\vec{x}$ naturally induces a partition  of $\{1,2,\cdots,2k\}$: $i$ and $j$ belong to the same block if and only if $x_i=x_j$. We denote this partition by $P(\vec{x})$. 

We now introduce some notation associated with a partition $P$. $|P|$ denotes the number of blocks in $P$. Let $L:=\{1,2,\cdots,k\}$ and $R:=\{k+1,k+2,\cdots,2k\}$.  For a block $B\in P$, define $l_B:=|B\cap L|$, $r_B:=|B\cap R|$, $q_B:=l_B-r_B$.

The key observation is that both $G(\vec{n},\vec{n}',T)$ and $M(\vec{n},\vec{n}')$ depend on $\vec{n}$ and $\vec{n}'$ only through the induced partition $P(\vec{x})$.  Let the blocks of $P(\vec{x})$ be $B_1,\cdots,B_{|P|}$. Since the GUE eigenvalue distribution $P(E_1,\cdots,E_D)$ is symmetric under permutations of  $E_1,\cdots,E_D$, we have
\begin{equation}
    g_{\vec{n},\vec{n}'}(\tau)=\int DE\exp\left[ i\tau (q_{B_1}E_1+\cdots+q_{B_{|P|}}E_{|P|}) \right].
\end{equation}
Therefore $g_{\vec{n},\vec{n}'}(\tau)$ and $G(\vec{n},\vec{n}',T)$ both only depend on the partition $P(\vec{x})$. We denote the corresponding functions by $g_P(\tau)$ and $G_P(T)$. To analyze $M(\vec{n},\vec{n}')$, we perform the Weingarten integration. This gives
\begin{equation}
M(\vec{n},\vec{n}')=\mathop{\mathbb{E}}_{\ket{\psi},\ket{\phi}\sim \mathcal{E}'}\sum _{\pi,\sigma \in S_{2k}}\mathrm{Wg}(\pi ^{-1}\sigma,D)\mathrm{tr}[V_{D}(\sigma ^{-1})(\ket{\phi}\bra{\phi} ^{\otimes k}\otimes \ket{\psi}\bra{\psi} ^{\otimes k})]\bra{\vec{n},\vec{n}'} V_{D}(\pi)\ket{\vec{n}',\vec{n}}.
\label{eq:defMnn}
\end{equation}
In this expression, the only dependence on $\vec{n}$ and $\vec{n}'$ appears in $\bra{\vec{n},\vec{n}'}V_D(\pi)\ket{\vec{n}',\vec{n}}$, which is either $0$ or $1$, depending on which of the \(x_i\)'s are equal to one another
and not on their specific values. Therefore $M(\vec{n},\vec{n}')$ also depends on $\vec{n}$ and $\vec{n}'$ only through the induced partition $P(\vec{x})$. We denote the corresponding  function by $M_P$.

Under this observation, the sum over $\vec{n}$ and $\vec{n}'$ in Eq.~\eqref{eq:sm_GUEstatefp} can be reorganized as a summation over all partitions:
\begin{equation}
    \mathop{\mathbb{E}}_{H\sim \mathrm{GUE}}F^{(k)}_{\mathcal{E}}=\sum_{P}D(D-1)\cdots(D-|P|+1)M_PG_P(T),
\end{equation}
where $D(D-1)\cdots(D-|P|+1)$ counts the number of $\vec{x}$ with $P(\vec{x})=P$. 

Now consider the partitions $P$ with all $q_{B}=0$, namely all blocks $B\in P$ have equal support on $\{ 1,\cdots,k \}$ and $\{ k+1, \cdots,2k\}$.  Note that $P(\vec{x})$ satisfies this condition if and only if $\vec{n}$ equals to $\vec{n}'$ up to a permutation, corresponding to the $k$-th no-resonance condition we used previously. For these partitions, $G_{P}(T)=1$. The sum over these partitions yields the $T\to \infty$ contribution of $\mathop{\mathbb{E}}_{H\sim \mathrm{GUE}}F^{(k)}_{\mathcal{E}}$ obtained in Theorem~1 in the main text. Therefore, the finite-$T$  correction comes from the remaining partitions, namely those with at least one block satisfying $q_B\neq 0$. 

In the following, we analyze $G_P(T)$ and $M_P$, respectively,  at leading order in $D$. We first prove that, for partitions $P$ with at least one nonzero $q_B$, $G_P(T)=O(1/T)$. Then we prove that, for all partitions $P$, $D(D-1)\cdots(D-|P|+1)M_P=O(F^{(k)}_{\mathcal{E}'})$. Combining these two estimates shows that the finite-$T$ correction to the GUE-averaged frame potential is $O(F^{(k)}_{\mathcal{E}'}/T)$.

\subsection{Analysis of $G_P(T)$}
We first analyze the behavior of the function $g_P(\tau)$. Let $P=\{B_1,\cdots,B_l\}$ be a partition with $l=|P|$ blocks. By definition,
\begin{equation}
    g_P(\tau)=\int DE\exp\left[ i\tau (q_{B_1}E_1+\cdots+q_{B_{l}}E_{l}) \right].
\end{equation}

We adopt the GUE convention used in Ref.~\cite{Cotler_2017}. For a $D\times D$ GUE matrix, its diagonal components are independent real Gaussian random variables $\mathcal{N}(0,1/{D})$ and its off-diagonal components are independent complex Gaussian random variables $\mathcal{CN}(0,1/{D})$. With this normalization, the spectral density converges to the Wigner semicircle distribution supported on $[-2,2]$. The joint probability density function of eigenenergies is
\begin{equation}
    P(E_{1},\cdots ,E_{D})=\frac{D^{D^{2}/2}}{(2\pi)^{D/2}\prod _{p=1}^{D}p!}\exp\left[ -\frac{D}{2}\sum _{i}E _{i}^{2} \right]\prod _{1\leq i<j\leq D}(E _{i}-E _{j})^{2},
\end{equation}
and the $n$-point spectral correlation function is defined as
\begin{equation}
    \rho ^{(n)}(E_{1},\cdots ,E_{n}):=\int dE_{n+1}\cdots dE_{D}P(E_{1},\cdots ,E_{D}).
\end{equation}
This function can be compactly expressed as 
\begin{equation}
    \rho ^{(n)}(E_{1},\cdots ,E_{n})=\frac{(D-n)!}{D!}\det[K_D^{(n)}(E_1,\cdots,E_n)],
\end{equation}
where $K_D^{(n)}(E_1,\cdots,E_n)$ is an $n\times n$ matrix. To write out its elements, define 
\begin{equation}
    K_D(x,x'):=\sum_{\alpha=0}^{D-1}\psi_\alpha(x)\psi_{\alpha}(x'),
    \label{eq:sm_KD}
\end{equation}
where $\psi_{\alpha}$ $(\alpha=0,1,\cdots,D-1)$ are $D$ orthonormal functions satisfying
\begin{equation}
    \int dx \,\psi_{\alpha}(x)\psi_{\beta}(x)=\delta_{\alpha\beta}.
\end{equation}
Refer to Ref.~\cite{fyodorov2010introductionrandommatrixtheory} for the details of $\psi_{\alpha}$. The matrix elements of $K^{(n)}_D(E_1,\cdots,E_n)$ are
\begin{equation}
[K^{(n)}_D(E_1,\cdots,E_n)]_{ij}=K_D(E_i,E_j).
\end{equation}

Under these definitions, $g_P(\tau)$ can be written as
\begin{equation}
    \begin{aligned}
        g_{P}(\tau)&=\int dE_{1}\cdots dE_{l}\,\rho ^{(l)}(E_{1},\cdots ,E_{l})\exp[i\tau (q_{B_1}E_1+\cdots+q_{B_{l}}E_{l})]\\
        &=\frac{(D-l)!}{D!}\sum_{\sigma\in S_l}\mathrm{sgn}(\sigma)\int dE_1\cdots dE_l\,K_D(E_1,E_{\sigma(1)})\cdots K_D(E_l,E_{\sigma(l)})\exp[i\tau (q_{B_1}E_1+\cdots+q_{B_{l}}E_{l})].
    \end{aligned}
\end{equation}
For simplicity, let \(s_a=q_{B_a}\tau\) for \(a=1,\ldots,l\). We denote the contribution from a fixed permutation \(\sigma\in S_l\) by
\begin{equation}
    I_\sigma(s_1,\ldots,s_l)
    :=
    \int dE_1\cdots dE_l\,
    \prod_{a=1}^{l}K_D(E_a,E_{\sigma(a)})
    \exp\left(i\sum_{a=1}^{l}s_aE_a\right).
    \label{eq:sm_Isigma}
\end{equation}
Then
\begin{equation}
    g_P(\tau)
    =
    \frac{(D-l)!}{D!}
    \sum_{\sigma\in S_l}
    \mathrm{sgn}(\sigma)\,
    I_\sigma(s_1,\ldots,s_l).
\end{equation}

We now analyze $I_{\sigma}(s_1,\cdots,s_l)$ for a fixed $\sigma$. Let \(\sigma=c_1\cdots c_p\) be the decomposition of \(\sigma\) into cycles. Then the integral in Eq.~\eqref{eq:sm_Isigma} can be  decomposed correspondingly into $p$ integrals. Specifically, consider one cycle \(c=(a_1a_2\cdots a_r)\), whose contribution to $I_{\sigma}$ is
\begin{equation}
    \begin{aligned}
    I_c
    &:=
    \int \prod_{j=1}^{r}dE_{a_j}\,
    K_D(E_{a_1},E_{a_2})K_D(E_{a_2},E_{a_3})
    \cdots K_D(E_{a_r},E_{a_1})
    \exp\left(i\sum_{j=1}^{r}s_{a_j}E_{a_j}\right),
    \end{aligned}
\end{equation}
and we have $I_{\sigma}(s_1,\cdots,s_l)=\prod_{a=1}^p I_{c_a}$.
Using Eq.~\eqref{eq:sm_KD}, we obtain
\begin{equation}
    \begin{aligned}
    I_c
    =
    \sum_{\alpha_1,\ldots,\alpha_r=0}^{D-1}
    &\left[
    \int dE_{a_1}\,
    \psi_{\alpha_r}(E_{a_1})e^{is_{a_1}E_{a_1}}\psi_{\alpha_1}(E_{a_1})
    \right] \\
    &\times
    \left[
    \int dE_{a_2}\,
    \psi_{\alpha_1}(E_{a_2})e^{is_{a_2}E_{a_2}}\psi_{\alpha_2}(E_{a_2})
    \right]\cdots \\
    &\times
    \left[
    \int dE_{a_r}\,
    \psi_{\alpha_{r-1}}(E_{a_r})e^{is_{a_r}E_{a_r}}\psi_{\alpha_r}(E_{a_r})
    \right].
    \end{aligned}
\end{equation}
This naturally leads us to define a \(D\times D\) matrix \(F(s)\) by
\begin{equation}
    F(s)_{\alpha\beta}
    :=
    \int dE\,\psi_\alpha(E)e^{isE}\psi_\beta(E),
    \qquad
    0\leq \alpha,\beta\leq D-1 .
\end{equation}
Then the contribution of the cycle \(c=(a_1a_2\cdots a_r)\) can be written as
\begin{equation}
    I_c
    =
    \mathrm{Tr}\!\left[
        F(s_{a_1})F(s_{a_2})\cdots F(s_{a_r})
    \right].
\end{equation}
Therefore,
\begin{equation}
    I_\sigma(s_1,\ldots,s_\ell)
    =
    \prod_{b=1}^p
    \mathrm{Tr}\!\left[
        \prod_{a\in c_b}F(s_a)
    \right],
\end{equation}
where the product inside each trace is ordered according to the corresponding cycle.

To bound $I_{\sigma}$, we first use Schatten \(\infty\)-norm to bound the matrix $F(s)$. We claim that \(\|F(s)\|_\infty\leq 1\) for all real \(s\). To see this, take any vector \(v\in\mathbb C^D\) such that $\|v\|_2\leq 1$ and define \(f_v(E):=\sum_{\beta=0}^{D-1}v_\beta\psi_\beta(E)\). Since the functions \(\{\psi_\beta\}_{\beta=0}^{D-1}\) are orthonormal, we have \(\|f_v\|_{L^2}=\|v\|_2\). Moreover,
\begin{equation}
    (F(s)v)_\alpha
    =
    \int dE\,\psi_\alpha(E)e^{isE}f_v(E).
\end{equation}
The right-hand side is the projection coefficient of the function \(e^{isE}f_v(E)\) onto \(\psi_\alpha\). Since \(\{\psi_\alpha\}_{\alpha=0}^{D-1}\) is only a finite orthonormal set, not a complete basis, the sum of the squared projection coefficients cannot exceed the full \(L^2\)-norm. This is precisely Bessel's inequality. Therefore,
\begin{equation}
    \|F(s)v\|_2^2=
    \sum_{\alpha=0}^{D-1}
    \left|
    \int dE\,\psi_\alpha(E)e^{isE}f_v(E)
    \right|^2\leq
    \|e^{isE}f_v(E)\|^2_{L^2}
    =
    \|f_v(E)\|^2_{L^2}
    =
    \|v\|_2^2 \leq 1.
\end{equation}
Thus \(\|F(s)\|_\infty\leq 1\) for all $s$. It follows that
\begin{equation}
    |I_{\sigma}(s_1,\cdots,s_l)|=\prod_{b=1}^p \left|\mathrm{Tr}\!\left[
        \prod_{a\in c_b}F(s_a)
    \right]\right|\leq \prod_{b=1}^p D\left\|
        \prod_{a\in c_b}F(s_a) 
    \right\|_\infty\leq \prod_{b=1}^p D\prod_{a\in c_b}\|F(s_a)\|_{\infty}\leq D^p
\end{equation}
and
\begin{equation}
    \frac{(D-l)!}{D!}|I_{\sigma}(s_1,\cdots,s_l)|\leq \frac{(D-l)!D^p}{D!}.
\end{equation}
Note that $p$ is the number of cycles in $\sigma$. Unless $\sigma=\mathrm{id}$, its contribution to $g_P(\tau)$ will vanish in the thermodynamic limit. 

We now analyze the $\sigma=\mathrm{id}$ contribution, which is
\begin{equation}
    I_{\mathrm{id}}(s_1,\cdots,s_l)=\prod_{a=1}^l \int dE\, K_D(E,E)e^{is_aE}.
\end{equation}
In the large $D$ limit, $K_D(E,E)\approx(D/2\pi)\sqrt{4-E^2}$, i.e., the Wigner semicircle distribution. Performing the integral, we get
\begin{equation}
    I_{\mathrm{id}}(s_1,\cdots,s_l)\approx D^l\prod _{a=1}^{l}\frac{J_1(2q_{B_a}\tau)}{q_{B_a}\tau}.
\end{equation}
where $J_1(x)$ denotes the Bessel function of the first kind. And therefore
\begin{equation}
    g_P(\tau)\approx \frac{(D-l)!D^l}{D!}\prod _{a=1}^{l}\frac{J_1(2q_{B_a}\tau)}{q_{B_a}\tau}.
\end{equation}
This approximation is consistent with the fact that $g_P(0)=1$ to leading order of $D$ since $\lim_{x\to 0}J_1(x)/x=1/2$. It shows that $g_P(\tau)$ starts from $1$, and then decays into oscillations with vanishing amplitude. The decay time can be estimated from the first few zeros of $J_1(x)$. Note that this decay time is independent of $D$, owing to our normalization of the GUE, for which the Wigner semicircle distribution is supported on $[-2,2]$. For local Hamiltonians, by contrast, the spectral width generally grows with $N$, so the decay time may decrease with $N$, as argued in the main text. This difference is not important when we only keep the
leading-order scaling in \(D\). Recall the definition of $G_P(T)$:
\begin{equation}
    G_P(T)=\frac{2}{T}\int_0^T d\tau \left(1-\frac{\tau}{T}\right)g_P(\tau).
\end{equation}
As long as $T$ is much larger than the decay time of $g_P(\tau)$, $G_P(T)$ will be of order $O(1/T)$.

\subsection{Analysis of $M_P$}
We first recall the definition of $M_P$, which is  equal to  $M(\vec{n},\vec{n}')$ for any $\vec{x}=(\vec{n},\vec{n}')$ with $P(\vec{x})=P$. 
\begin{equation}
M_P=M(\vec{n},\vec{n}')=\mathop{\mathbb{E}}_{\ket{\psi},\ket{\phi}\sim \mathcal{E}'}\sum _{\pi,\sigma \in S_{2k}}\mathrm{Wg}(\pi ^{-1}\sigma,D)\mathrm{tr}[V_{D}(\sigma ^{-1})(\ket{\phi}\bra{\phi} ^{\otimes k}\otimes \ket{\psi}\bra{\psi} ^{\otimes k})]\bra{\vec{n},\vec{n}'} V_{D}(\pi)\ket{\vec{n}',\vec{n}}, 
\label{eq:sm_defMP}
\end{equation}

Here we introduce some notation. Let  $P=\{ B_{1},\cdots,B_{l} \}$ be a partition with $l=|P|$ blocks. Pick a specific element in the permutation group $S_{2k}$: $s:=(1,k+1)(2,k+2)\cdots(k,2k)$.  For a permutation $\sigma\in S_{2k}$, define $a(\sigma):=|L\cap \sigma(L)|$, which is also equal to $|R\cap\sigma(R)|$. Recall that $L=\{1,2,\cdots,k\}$ and $R=\{k+1,k+2,\cdots,2k\}$. Then $a(\sigma)$ counts how many elements remain in $L$($R$) after the action of $\sigma$.

Now focus on
\begin{equation}
     \bra{\vec{n},\vec{n}'} V_{D}(\pi)\ket{\vec{n}',\vec{n}} = \bra{\vec{n},\vec{n}'} V_{D}(\pi s^{-1})\ket{\vec{n},\vec{n}'},
\end{equation}
which is either $0$ or $1$. And it equals $1$ if and only if $\pi s^{-1}$ only permutes elements within  each block $B_i$. These permutations constitute a subgroup of $S_{2k}$, denoted by $H_P$. Constraining the sum over $\pi$ in Eq.~\eqref{eq:sm_defMP} to this subgroup, we get
\begin{equation}
    \begin{aligned}
M_P=M(\vec{n},\vec{n}')&=\mathop{\mathbb{E}}_{\ket{\psi},\ket{\phi}\sim \mathcal{E}'}\sum _{\pi,\sigma \in S_{2k}}\mathrm{Wg}(\pi ^{-1}\sigma,D)\bra{\phi ^{k}\psi ^{k}} V_{D}(\sigma ^{-1})\ket{\phi ^{k}\psi ^{k}} \bra{\vec{n},\vec{n}'} V_{D}(\pi s ^{-1})\ket{\vec{n},\vec{n}'}\\
&=\mathop{\mathbb{E}}_{\ket{\psi},\ket{\phi}\sim \mathcal{E}'}\sum _{\sigma \in S_{2k}}\sum _{h\in H_{P}}\mathrm{Wg}(s ^{-1}h^{-1}\sigma,D)\bra{\phi ^{k}\psi ^{k}} V_{D}(\sigma ^{-1})\ket{\phi ^{k}\psi ^{k}} \\
&=\mathop{\mathbb{E}}_{\ket{\psi},\ket{\phi}\sim \mathcal{E}'}\sum _{h\in H_{P}}\sum _{\omega \in S_{2k}}\mathrm{Wg}(\omega ^{-1},D)\bra{\phi ^{k}\psi ^{k}} V_{D}(\omega s ^{-1}h^{-1})\ket{\phi ^{k}\psi ^{k}}\\
&=\mathop{\mathbb{E}}_{\ket{\psi},\ket{\phi}\sim \mathcal{E}'}\sum _{h\in H_{P}}\sum _{\omega \in S_{2k}}\mathrm{Wg}(\omega ^{-1},D)\bra{\phi ^{k}\psi ^{k}} V_{D}(\omega s ^{-1}h^{-1}s)\ket{\psi ^{k}\phi ^{k}}\\
&=\sum _{h\in H_{P}}\sum _{\omega \in S_{2k}}\mathrm{Wg}(\omega ^{-1},D)F^{(a(\omega \tilde{h}))}_{\mathcal{E}'}
\end{aligned}
\end{equation}
In the second line, we let $h=\pi s^{-1}$. In the third  line, we let $\omega=\sigma^{-1}hs$. In the fifth line, we denote $\tilde{h}:=s^{-1}hs$ (we change $h^{-1}$ to $h$ since $H_P$ is a group) and use the identity
\begin{equation}
    \mathop{\mathbb{E}}_{\ket{\psi},\ket{\phi}\sim \mathcal{E}'}\bra{\phi ^{k}\psi ^{k}} V_{D}(\sigma)\ket{\psi ^{k}\phi ^{k}}=F^{(a(\sigma))}_{\mathcal{E}'}.
\end{equation}

We first consider the contribution from $\omega=\mathrm{id}$, in which case $\mathrm{Wg}(\mathrm{id},D)=O(1/D^{2k})$ and the frame potential term becomes $F^{(a(\tilde{h}))}_{\mathcal{E}'}$. It  can be proved that $a(\tilde{h})=a(s^{-1}hs)=a(h)$. Indeed,
\begin{equation}
    a(s^{-1}hs)=|L\cap s^{-1}hs(L)|=|L\cap s^{-1}h(R)|=|s(L)\cap h(R)|=|R\cap h(R)|=a(h).
\end{equation}
To extract the leading order contribution of $\sum_{h\in H_P}\mathrm{Wg}(\mathrm{id},D)F_{\mathcal{E}'}^{(a(h))}$, we need to find out which elements $h\in H_P$ minimize $a(h)$ since the state frame potential is non-increasing in its order. Recall that $h$ only permutes elements within each block of $P$ and  $a(h)$ means how many elements of $L$ remain in $L$ after the action of $h$. If a block $B$ is entirely contained in $L$ or $R$, then permutations inside $B$ do not help move elements from $L$ to $R$. However, if  $B$ has both support on $L$ and $R$, we can move at most $\mathrm{min}(l_B,r_B)$ elements from $L$ to $R$. Therefore,
\begin{equation}
    a_{\mathrm{min}}(P):=\mathop{\mathrm{min}}_{h\in H_P}a(h)=k-\sum_{B\in P}\mathrm{min}(l_B,r_B).
\end{equation}
And the leading order contribution to $D(D-1)\cdots(D-|P|+1)M_P$ from $\omega=\mathrm{id}$ is
\begin{equation}
    O\left(\frac{F^{(a_{\mathrm{min}}(P))}_{\mathcal{E}'}}{D^{2k-|P|}}\right).
    \label{eq:sm_omegaid}
\end{equation}

We now compare Eq.~\eqref{eq:sm_omegaid} with $F^{(k)}_{\mathcal{E}'}$. Since for each block $B$, $|B|-1\geq \mathrm{min}(l_B,r_B)$, we have
\begin{equation}
    2k-|P|=\sum_{B\in P}(|B|-1)\geq \sum_{B\in P}\mathrm{min}(l_{B},r_{B})=k-a_{\mathrm{min}}(P).
\end{equation}
Thus
\begin{equation}
    O\left(\frac{F^{(a_{\mathrm{min}}(P))}_{\mathcal{E}'}}{D^{2k-|P|}}\right)\leq O\left(\frac{F^{(a_{\mathrm{min}}(P))}_{\mathcal{E}'}}{D^{k-a_{\mathrm{min}}(P)}}\right)\leq O(F^{(k)}_{\mathcal{E}'}).
\end{equation}
The second inequality follows from the fact that for any state ensemble $\mathcal{E}$, we have $F^{(m+1)}_{\mathcal{E}}\geq F^{(m)}_{\mathcal{E}}/D$. Intuitively, this is because the frame potential may decrease in its order, but the rate at which it decreases is at most $O(1/D)$, since the Haar ensemble has the smallest frame potential and decrease with the rate $O(1/D)$. If the frame potential of some ensemble decreased faster than $O(1/D)$, it would become smaller than the Haar value at some order $m$. We now prove this relation rigorously. For any state ensemble $\mathcal{E}$, we have such inequality:
\begin{equation}
    \begin{aligned}(F^{(k)}_{\mathcal{E}})^{2}&=\left(\mathop{\mathbb{E}}_{\ket{\psi} ,\ket{\phi} \sim \mathcal{E}}|\langle \psi|\phi \rangle |^{2k}\right)^{2}=\left(\mathop{\mathbb{E}}_{\ket{\psi} ,\ket{\phi} \sim \mathcal{E}}|\langle \psi|\phi \rangle |^{(k-1)}|\langle \psi|\phi \rangle |^{(k+1)}\right)^{2}\\
    &\leq \left(  \mathop{\mathbb{E}}_{\ket{\psi} ,\ket{\phi} \sim \mathcal{E}}|\langle \psi|\phi \rangle |^{2(k-1)}\right) \left(  \mathop{\mathbb{E}}_{\ket{\psi} ,\ket{\phi} \sim \mathcal{E}}|\langle \psi|\phi \rangle |^{2(k+1)}\right)=F^{(k-1)}_{\mathcal{E}}F^{(k+1)}_{\mathcal{E}},
    \end{aligned}
\end{equation}
where we use the Cauchy-Schwarz inequality. This inequality implies
\begin{equation}
    \frac{F^{(k+1)}_{\mathcal{E}}}{F^{(k)}_{\mathcal{E}}}\geq\frac{F^{(k)}_{\mathcal{E}}}{F^{(k-1)}_{\mathcal{E}}}\geq \cdots \geq F^{(1)}_{\mathcal{E}}\geq F^{(1)}_{\mathrm{Ha ar}}=1/D.
\end{equation}

It remains to consider the contribution from $\omega\neq \mathrm{id}$. We need to bound $\mathrm{Wg}(\omega ^{-1},D)F^{(a(\omega \tilde{h}))}_{\mathcal{E}'}$. Let $\#\omega$ denote the number of cycles in $\omega$ and define $|\omega|:= 2k-\#\omega$, which is the minimal number of swaps to express $\omega$. The Weingarten function satisfies $\mathrm{Wg}(\omega^{-1},D)=O(D^{-2k-|\omega|})$. For a single swap $\tau$, we have
\begin{equation}
    a(\tau \tilde{h})\geq a(\tilde{h})-1
\end{equation}
since a single swap can at most move one element from $L$ to $R$. Using this inequality $|\omega|$ times, we get
\begin{equation}
    a(\omega \tilde{h})\geq a(\tilde{h})-|\omega|.
\end{equation}
Thus we have
\begin{equation}
    \mathrm{Wg}(\omega ^{-1},D)F^{(a(\omega \tilde{h}))}_{\mathcal{E}'}\leq O(F_{\mathcal{E}'}^{(a(\tilde{h})-|\omega|)}/D^{2k+|\omega|})\leq O(F_{\mathcal{E}'}^{(a(\tilde{h}))}/D^{2k})
\end{equation}
and we reduce to the $\omega=\mathrm{id}$ case.

\section{Stabilizer state initial ensemble}
\label{ap:stabilizer}
In this section, we consider the ensemble $\mathcal{E}=\{e^{-iHt}\ket{\psi}\mid t\sim \mathrm{Unif}[0,T],\ket{\psi}\sim \mathcal{E}' \}$ and take the initial ensemble $\mathcal{E}'$ to be the random stabilizer-state ensemble. Assuming $H$ satisfies the $k$-th no-resonance condition and taking the limit $T\to \infty$, we prove that the evolved ensemble $\mathcal{E}$ forms a state $k$-design when the eigenbasis of $H$ is a Clifford rotation of the computational basis, i.e., $\ket{E_{n}}=C\ket{n}$, where $\{\ket{n}\}$ denotes the computational basis and $C$ is an arbitrary Clifford unitary. We also prove that, in order to construct a computational-basis diagonal Hamiltonian satisfying the $k$-th no-resonance condition, one needs $\Omega(\log k)$-body interactions.

First we review several facts about the $k$-th moment operator of the random stabilizer-state ensemble~\cite{Zhang_2026,Gross_2021,bittel2025completetheorycliffordcommutant}. Throughout this section  we assume the number of qubits $N\geq k-1$. The $k$-th moment operator of the random stabilizer-state ensemble can be written as
\begin{equation}
    \rho ^{(k)}_{\mathcal{E}'}=\frac{1}{Z_{N}}\sum _{\Omega \in \Sigma_{k,k}}\Omega, 
\end{equation}
where $Z_{N}=D\prod _{i=0}^{k-2}(D+2^{i})=D^{k}+O(D^{k-1})$ and $\Sigma_{k,k}$ is a linearly independent  basis of the Clifford commutant.

The elements of $\Sigma_{k,k}$ can be described as follows. Let $T\subset \mathbb{F}_2^{2k}$ be a subspace and  label the elements of $T$ as $(x,y)$, with $x,y\in\mathbb{F}_2^k$. $T$ is a Lagrangian subspace if it satisfies the following three conditions: (1) $|x|=|y|(\mathrm{mod} \,4)$ for all $(x,y)\in T$; (2) $\mathrm{dim}(T)=k$; (3) $(1,1,\cdots,1)\in T$. For each Lagrangian subspace, we can define an operator
\begin{equation}
    R(T)=\left(\sum_{(x,y)\in T}\ket{x}\bra{y}\right)^{\otimes N}.
\end{equation}
The set of all such operators forms $\Sigma_{k,k}$, whose size is $|\Sigma_{k,k}|=\prod _{i=0}^{k-2}(2^{i}+1)$, independent of $N$. 

For $k\leq3$, we have $|\Sigma _{k,k}|=k!$ and the elements of $\Sigma _{k,k}$ are precisely the permutation operators. For $k>3$, however,  $|\Sigma _{k,k}|>k!$, which means the Clifford commutant contains additional elements beyond permutations. For example, $|\Sigma _{4,4}|=30$. Besides the $24$ permutations, there are 6 additional elements: $\Omega _{4}$, $\Omega _{4}V_{D}((12))$, $\Omega _{4}V_{D}((13))$, $\Omega _{4}V_{D}((14))$, $\Omega _{4}V_{D}((123))$, $\Omega _{4}V_{D}((132))$, where $\Omega _{4}=(1/D)\sum_{P\in \mathcal{P}_N}P^{\otimes 4}$. This reflects the well-known  fact that the set of stabilizer states forms a state $3$-design, but not a $4$-design.

We will also use the following trace properties of the elements in $\Sigma _{k,k}$. An element $\Omega \in\Sigma_{k,k}$ satisfies $\mathrm{tr}(\Omega)=D^k$ if and only if $\Omega$ is the identity operator; otherwise $\mathrm{tr}(\Omega)=O(D^{k-1})$. Moreover,
\begin{equation}
    \mathrm{tr}(\Omega ^{\dagger}\Omega')\begin{cases}
=D^{k},&\Omega=\Omega';\\ \\
\leq D^{k-2},&\Omega \neq \Omega'.
\end{cases}
\end{equation}
Therefore, the frame potential of the stabilizer states is
\begin{equation}
    F^{(k)}_{\mathcal{E}'}=\mathrm{tr}[(\rho ^{(k)}_{\mathcal{E}'})^2]=\frac{1}{Z_N^2}\sum_{\Omega,\Omega'\in \Sigma_{k,k}}\mathrm{tr}(\Omega^\dagger\Omega')=\frac{|\Sigma_{k,k}|}{D^k}+O(1/D^{k+1}).
\end{equation}

Now we analyze the frame potential of the evolved ensemble. Starting from Eq.~\eqref{eq:sm_derive_eq2} and using \(1/\Omega(\vec n)\leq 1\), we obtain
\begin{equation}
    \begin{aligned}
        F^{(k)}_{\mathcal{E}}
        &\leq
        k!\sum_{\vec n}
        \left(
        \bra{E_{n_1}\cdots E_{n_k}}
        \rho^{(k)}_{\mathcal{E}'}
        \ket{E_{n_1}\cdots E_{n_k}}
        \right)^2 \\
        &=
        \frac{k!}{Z_N^2}
        \sum_{\Omega,\Omega'\in\Sigma_{k,k}}
        \sum_{\vec n}
        \bra{E_{n_1}\cdots E_{n_k}}\Omega
        \ket{E_{n_1}\cdots E_{n_k}}
        \bra{E_{n_1}\cdots E_{n_k}}\Omega'
        \ket{E_{n_1}\cdots E_{n_k}} .
    \end{aligned}
\end{equation}
We first take the eigenbasis of \(H\) to be the computational basis, i.e.,
\(\ket{E_n}=\ket n\). In this case, the leading contribution to the above upper bound comes only from
\(\Omega=\Omega'=\mathrm{id}\). To see this, note that, by the definition of \(\Sigma_{k,k}\), the matrix elements of every \(\Omega\in\Sigma_{k,k}\) in the computational basis are either \(0\) or \(1\). Therefore,
\begin{equation}
    \sum_{\vec n}
    \bra{n_1\cdots n_k}\Omega\ket{n_1\cdots n_k}
    \bra{n_1\cdots n_k}\Omega'\ket{n_1\cdots n_k}
    \leq
    \min\{\mathrm{tr}(\Omega),\mathrm{tr}(\Omega')\}.
\end{equation}
Using the  trace properties of $\Omega$, we obtain $F^{(k)}_{\mathcal{E}}\leq k!/D^k+O(1/D^{k+1})$. And  since the Haar ensemble has the smallest frame potential with the leading order $k!/D^k$, the leading order of $F^{(k)}_{\mathcal{E}}$ has to be $k!/D^k$.

A direct corollary is that we can also choose the eigenstates to be $\ket{E_n}=C\ket{n}$, where $C$ is an arbitrary Clifford unitary, since $C^{\dagger\otimes k} \Omega C^{\otimes k}=\Omega$ is true for all Clifford $C$ and all $\Omega\in \Sigma_{k,k}$.

We now ask what conditions are required for a computational-basis diagonal Hamiltonian to satisfy the $k$-th no-resonance condition. Note that such Hamiltonians can be written only with Pauli matrices $I$ and $Z$. If we demand the Hamiltonian to be local in 1-D geometry and only contain few-body interactions, we can never achieve the non-resonance condition of order $2$. If we discard locality and insist on few-body interactions, a general $r$-body Hamiltonian can be written in the form
\begin{equation}
    H=\sum_{|S|\leq r}J_S Z_S
\end{equation}
where $S$ are subsets of $[N]:=\{1,2,\cdots,N\}$ and $Z_S=\prod_{i\in S}Z_i$. In fact, we have the following theorem:
\begin{theorem}
  Consider Hamiltonians whose eigenbasis is the computational basis, which we refer to as classical Hamiltonians. No $r$-body Hamiltonian can satisfy the $2^r$-th no-resonance condition. Conversely, there exist $r$-body Hamiltonians satisfying the $k$-th no-resonance condition for any $k<2^r$. This implies, to achieve the $k$-th  no-resonance condition, we need a $(\lfloor \log_2k\rfloor+1)$-body Hamiltonian.
\end{theorem}
\begin{proof}
    To prove the theorem, we need some facts about Boolean functions~\cite{odonnell2021analysisbooleanfunctions}. A real-valued Boolean function is a function $f:\{ -1,1 \}^{N}\to \mathbb{R}$. All such functions form a real vector space with dimension $2^{N}$, which we denote by $V_N$. We equip $V_N$ with the inner product:
    \begin{equation}
        \langle f,g\rangle=2^{-N}\sum _{x \in \{ -1,1 \}^{N}}f(x)g(x).
    \end{equation}
    For each subset $S\subset [N]$, define
    \begin{equation}
        \chi _{S}(x):=\prod _{i\in S}x_{i},
    \end{equation}
    which is a function in $V_N$. All such functions form an orthonormal basis of $V_{N}$. Therefore any function $f$ in $V_N$ can be expanded as 
    \begin{equation}
        f=\sum_{S\subset [N]}\langle \chi_S,f\rangle \chi_S.
    \end{equation}
This expansion is called the Fourier expansion of $f$, and the coefficients $\langle \chi_S,f\rangle$ are called Fourier coefficients. Additionally, the degree of $f$ is defined as
$\mathrm{deg}(f)=\mathrm{max}_{\langle \chi_S,f\rangle\neq0}(|S|)$. The support of $f$ is defined as $\mathrm{supp}(f)=\{x\in \{-1,1\}^N\mid    f(x)\neq 0   \}$.

The eigenenergy of an $r$-body classical Hamiltonian can be written as
\begin{equation}
    E(x)=\sum _{|S|\leq r}J_{S}\chi _{S}(x)
    \label{eq:sm_rlocal}
\end{equation}
where $x\in \{-1,1\}^N$, which is a real-valued Boolean function with $\mathrm{deg}(E)\leq r$. We first prove that such a Hamiltonian cannot satisfy the $2^r$-th no-resonance condition. Choose a subset $A\subset [N]$ with $|A|=r+1$. Without loss of generality, we take $A=\{1,2,\cdots,r+1\}$. Now fix the values of all bits outside of $A$:
\begin{equation}
    x_i=s_i,\quad i=r+2,\cdots,N.
\end{equation}
This gives a restricted function $E_A\in V_{r+1}$ defined by
\begin{equation}
    E_A(z)=E(z,s_{r+2},\cdots,s_N),\quad z\in \{-1,1\}^{r+1}.
\end{equation}
Note that $\mathrm{deg}(E_A)\leq r$, which follows directly by substituting Eq.~\eqref{eq:sm_rlocal} into the definition of $E_A$. Therefore we have
\begin{equation}
    \langle \chi _{A},E_{A}\rangle=2^{-r-1}\sum _{z \in \{ -1,1 \}^{r+1}}\chi _{A}(z)E_{A}(z)=0.
    \label{eq:sm_no_r_plus_one}
\end{equation}
Here $\chi _{A}(z)=z_1\cdots z_{r+1}$ is a function in $V_{r+1}$. Use $\chi _{A}$ to divide $\{ -1,1 \}^{r+1}$ into two sets:
\begin{equation}
    P_{+}=\{ z\in \{ -1,1 \}^{r+1}\mid \chi _{A}(z)=1 \}\quad P_{-}=\{ z\in \{ -1,1 \}^{r+1}\mid \chi _{A}(z)=-1 \}.
\end{equation}
Clearly, $|P_{+}|=|P_{-}|=2^{r}$. Then Eq.~\eqref{eq:sm_no_r_plus_one} implies
\begin{equation}
    \sum _{z\in P_{+}}E_{A}(z)=\sum _{z\in P_{-}}E_{A}(z).
\end{equation}
Writing this equation with the original energy function $E$, we get
\begin{equation}
    \sum _{z\in P_{+}}E(z,s_{r+2},\cdots ,s_{N})=\sum _{z\in P_{-}}E(z,s_{r+2},\cdots ,s_{N}),
\end{equation}
which implies any $r$-body classical Hamiltonian can't satisfy the $2^r$-th no-resonance condition.

Now we prove the converse direction by contradiction. Fix $k<2^r$, and suppose that no $r$-body Hamiltonian satisfies the $k$-th no-resonance condition. In other words,  for any energy function $E$ defined in Eq.~\eqref{eq:sm_rlocal}, there exist $X=\{x^{(1)},\cdots,x^{(k)}\}$ and $Y=\{y^{(1)},\cdots,y^{(k)}\}$ such that
\begin{equation}
    E(x^{(1)})+\cdots +E(x^{(k)})=E(y^{(1)})+\cdots +E(y^{(k)})
    \label{eq:sm_no_reso}
\end{equation}
and $X\neq Y$ in the sense of multisets. Here, $x^{(1)},\cdots,x^{(k)},y^{(1)},\cdots,y^{(k)}\in \{-1,1\}^N$. For any $|S|\leq r$, define 
\begin{equation}
    \Delta_S(X,Y):=\sum_{a=1}^k\chi_S(x^{(a)})-\sum_{a=1}^k\chi_S(y^{(a)}).
\end{equation}
Then Eq.~\eqref{eq:sm_no_reso} can be written as
\begin{equation}
    \sum_{|S|\leq r}J_S \Delta_S(X,Y)=0,
    \label{eq:sm_JDel}
\end{equation}
and our assumption can be restated as follows: for every choice of the Hamiltonian parameters $J_S$, there exist $X\neq Y$ such that Eq.~\eqref{eq:sm_JDel} holds.

We regard the Hamiltonian coefficients as coordinates of a real vector space $W$ with dimension $d=\sum_{i=0}^r\binom{N}{i}$. For each fixed  pair $X\neq Y$, let $W_{(X,Y)}\subset W$ be the set of parameters satisfying  Eq.~\eqref{eq:sm_JDel}. If $\Delta_S(X,Y)$ are not all zero, $W_{(X,Y)}$ is a subspace of $W$ with dimension $d-1$. If $\Delta_S(X,Y)=0$ for all $|S|\leq r$, we have $W_{(X,Y)}=W$. The previous assumption amounts to saying that
\begin{equation}
    W=\bigcup_{X\neq Y}W_{(X,Y)}.
\end{equation}
A $d$-dimensional real vector space cannot be expressed as the union of finitely many subspaces whose dimensions are all strictly less than $d$. Thus there must exist a pair of $X\neq Y$ such that $W_{(X,Y)}=W$. Equivalently, for this pair $X\neq Y$, we have
\begin{equation}
    \Delta_S(X,Y)=\sum_{a=1}^k\chi_S(x^{(a)})-\sum_{a=1}^k\chi_S(y^{(a)})=0
\end{equation}
for any $|S|\leq r$.

Then we prove that if there exists a pair of $X\neq Y$ such that $\Delta_S(X,Y)=0$ for all $|S|\leq r$, we must have $k\geq 2^r$.  
Define a function $\mu \in V_{N}$ by
\begin{equation}
    \mu(z)=\#\{ a:x^{(a)}=z \}-\#\{ a:y^{(a)}=z \}.
\end{equation}
The two multisets $X$ and $Y$ are equal if and only if for any $z\in \{ -1,1 \}^{N}$,  $\mu(z)=0$. Since $X$ and $Y$ each contain $k$ elements, we also have $|\mathrm{supp}(\mu)|\leq 2k$. Using the definition of $\mu$, for any $|S|\leq r$, we have
\begin{equation}
    0=\sum _{a=1}^{k}\chi _{S}(x^{(a)})-\chi _{S}(y^{(a)})=\sum _{z\in \{ -1,1 \}^{N}}\mu(z)\chi _{S}(z)=2^N\langle \mu,\chi_S\rangle,
\end{equation}
which means that the Fourier coefficient of $\mu$ is zero for any $|S|\leq r$.   

Next, define a new function $g(x)=\chi _{[N]}(x)\mu(x)$, where $\chi _{[N]}(x)=\prod _{a=1}^{N}x_{a}$. Since $\chi_{[N]}\neq 0$ for any $x\in\{-1,1\}^N$, we have $\mathrm{supp}(g)=\mathrm{supp}(\mu)$. Moreover,  $\mathrm{deg}(g)\leq N-r-1$ since $\chi_S\cdot\chi_{[N]}=\chi_{[N]\setminus S}$. From Lemma 3.5 of Ref.~\cite{odonnell2021analysisbooleanfunctions}, we know
\begin{equation}
    |\mathrm{supp}(\mu)|=|\mathrm{supp}(g)|\geq 2^{r+1}
\end{equation}
Then we get $k\geq 2^r$, which contradicts the assumption $k<2^r$.
\end{proof}

A closely related logarithmic threshold appears in the study of diagonal-unitary designs~\cite{Nakata_2014}, where the authors proved that an \(r\)-qubit phase-random diagonal circuit forms an exact diagonal-unitary \(k\)-design if and only if $r\geq \lfloor \log_2 k\rfloor+1$.

\section{Proof of Theorem~3}
\label{ap:Theorem3}
In this section, we prove Theorem 3 in the main text.

We first consider the frame potential of the ensemble $\mathcal{E}_{L}=\{ e^{-iHt}V\mid t\sim \mathrm{Unif}[0,T],V\sim \mathcal{E}' \}$, which is
\begin{equation}
\mathcal{F}^{(k)}_{\mathcal{E}_L}=\mathop{\mathbb{E}}\limits_{V,V'\sim \mathcal{E}'}\,\mathop{\mathbb{E}}\limits_{t,t'} |\mathrm{Tr}[V'^{\dagger}e^{iHt'}e^{-iHt}V]|^{2k}.
\end{equation}
Assuming $T\to \infty$ and that the Hamiltonian $H$ satisfies the $k$-th no-resonance condition, we can perform the time average in the same way as in the proof of Theorem 1 and obtain
\begin{equation} \mathcal{F}^{(k)}_{\mathcal{E}_L}=\mathop{\mathbb{E}}\limits_{V,V'\sim \mathcal{E}'}k!\sum _{\vec{n}}\frac{1}{\Omega(\vec{n})}\bra{E_{n_{1}}\cdots E_{n_{k}}E_{n_{1}}\cdots E_{n_{k}}} (VV'^{\dagger})^{\otimes k}\otimes(V'V^{\dagger})^{\otimes k}\ket{E_{n_{1}}\cdots E_{n_{k}}E_{n_{1}}\cdots E_{n_{k}}}. 
\end{equation}
Averaging over the GUE eigenbasis by Weingarten calculus, we get
\begin{equation}
    \begin{aligned}
        \mathop{\mathbb{E}}\limits_{H\sim \mathrm{GUE}}\mathcal{F}^{(k)}_{\mathcal{E}_L}=&k!\sum _{\pi,\sigma \in S_{2k}}\mathrm{Wg}(\pi ^{-1}\sigma,D)\mathop{\mathbb{E}}\limits_{V,V'\sim \mathcal{E}'}\mathrm{Tr}[V^{\dagger}_{D}(\sigma)(VV'^{\dagger})^{\otimes k}\otimes(V'V^{\dagger})^{\otimes k}] \\
&\quad \times \sum _{\vec{n}}\frac{1}{\Omega(\vec{n})}\bra{n_{1}\cdots n_{k}n_{1}\cdots n_{k}}V_{D}(\pi)\ket{n_{1}\cdots n_{k}n_{1}\cdots n_{k}}.
    \end{aligned}
    \label{eq:sm_unitarygue}
\end{equation}
We now analyze the leading-order contribution of Eq.~\eqref{eq:sm_unitarygue}. To simplify notation, let $W:= VV'^{\dagger}$ and $ \mathbb{E}_{\mathcal{E}'}:=\mathop{\mathbb{E}}\limits_{V,V'\sim \mathcal{E}'} $. Define
\begin{equation}
    \begin{aligned}
        &T_{\sigma}(W):=\mathrm{Tr}[V^{\dagger}_{D}(\sigma)(W^{\otimes k}\otimes W^{\dagger \otimes k})],\\
        &S_\pi :=\sum _{\vec{n}}\frac{1}{\Omega(\vec{n})}\bra{n_{1}\cdots n_{k}n_{1}\cdots n_{k}}V_{D}(\pi)\ket{n_{1}\cdots n_{k}n_{1}\cdots n_{k}}.
    \end{aligned}
\end{equation}
Then Eq.~\eqref{eq:sm_unitarygue} becomes
\begin{equation}
    \mathop{\mathbb{E}}\limits_{H\sim \mathrm{GUE}}\mathcal{F}^{(k)}_{\mathcal{E}_L}=k!\sum _{\pi,\sigma \in S_{2k}}\mathrm{Wg}(\pi ^{-1}\sigma,D)\mathbb{E}_{\mathcal{E}'}[T_\sigma(W)]S_\pi.
    \label{eq:sm_ushort}
\end{equation}
First, summing all terms with $\pi=\sigma\in E_{2k}$ and keeping only the leading-order contribution,  we obtain
\begin{equation}
    k!\left( 1+\sum _{m=1}^{k}\binom{k}{m}\frac{\mathcal{F}^{(m)}_{\mathcal{E}'}}{D^m} \right).
    \label{eq:sm_leadunitary}
\end{equation}
Below we prove this is precisely the leading-order contribution of Eq.~\eqref{eq:sm_unitarygue}.

There are two cases. The first case is that $\sigma\in E_{2k}$ but $\pi\neq\sigma$. In this case, $\mathbb{E}_{\mathcal{E}'}[T_\sigma(W)]=D^{k-m}\mathcal{F}^{(m)}_{\mathcal{E}'}$ for some $m$, while $S_\pi$ is at most $O(D^k)$ and the Weingarten function is at most $O(D^{-2k-1})$ since $\pi\neq \sigma$. Therefore the total contribution of such a term is $O(\mathcal{F}^{(m)}_{\mathcal{E}'}/D^{m+1})$, which is smaller by at least  a factor of $1/D$ than the corresponding term in Eq.~\eqref{eq:sm_leadunitary}.

The second case is that $\sigma\notin E_{2k}$. Write the cycle decomposition of $\sigma$ as $\sigma=c_1c_2\cdots c_l$. For a cycle $c$, let $\#_L(c)$  denote the number of elements of $\{ 1,\cdots,k \}$ contained in $c$ and let $\#_R(c)$ denote the number of elements of  $\{ k+1,\cdots,2k \}$ contained in $c$. Then
\begin{equation}
    T_\sigma(W)=\prod _{i=1}^{l}\mathrm{tr}[W^{\#_{L}(c_{i})}W^{\dagger\#_{R}(c_{i})}].
\end{equation}
Let $M(\sigma)$ be the number of cycles in $\sigma$ satisfying $|\#_L(c)-\#_{R}(c)|=1$ (Cycles of length one, also called fixed points of $\sigma$, are examples of such cycles). Then we can bound the norm of $\mathbb{E}_{\mathcal{E}'}[T_\sigma(W)]$:
\begin{equation}
    |\mathbb{E}_{\mathcal{E}'}[T_\sigma(W)]|\leq \mathbb{E}_{\mathcal{E}'}|T_\sigma(W)|\leq D^{l-M(\sigma)}\mathbb{E}_{\mathcal{E}'}[|\mathrm{tr}W|^{M(\sigma)}],
\end{equation}
where we bound the contribution of cycles that do not satisfy $|\#_L(c)-\#_{R}(c)|=1$ by $D$. To make $|\mathrm{Wg}(\pi ^{-1}\sigma,D)\mathbb{E}_{\mathcal{E}'}[T_\sigma(W)]S_\pi|$ as large as possible, we can either choose $\pi=\sigma$ to maximize the Weingarten function or choose $\pi\in E_{2k}$ to maximize $S_\pi$. But since we have assumed $\sigma \notin E_{2k}$, these two maximizations cannot be achieved simultaneously, and we get
\begin{equation}
    |\mathrm{Wg}(\pi ^{-1}\sigma,D)\mathbb{E}_{\mathcal{E}'}[T_\sigma(W)]S_\pi|=O\left(\frac{\mathbb{E}_{\mathcal{E}'}[|\mathrm{tr}W|^{M(\sigma)}]}{D^{k+1-l+M(\sigma)}}\right).
\end{equation}

If $M(\sigma)$ is even, we have
\begin{equation}
    \frac{\mathbb{E}_{\mathcal{E}'}[|\mathrm{tr}W|^{M(\sigma)}]}{D^{k+1-l+M(\sigma)}}=\frac{\mathcal{F}^{(M(\sigma)/2)}_{\mathcal{E}'}}{D^{k+1-l+M(\sigma)}}.
    \label{eq:sm_even1}
\end{equation}
We need to compare it with the corresponding term in Eq.~\eqref{eq:sm_leadunitary}, which is
\begin{equation}
    \frac{\mathcal{F}^{(M(\sigma)/2)}_{\mathcal{E}'}}{D^{M(\sigma)/2}}.
    \label{eq:sm_even2}
\end{equation}
In fact we can derive a simple relation  between $M(\sigma)$, the number of cycles $l$ and $k$. Since all $M(\sigma)$ cycles in $\sigma$ have length at least $1$ while the remaining $l-M(\sigma)$ cycles have length at least $2$, we have $2k\geq M(\sigma)+2(l-M(\sigma))$, or equivalently, 
\begin{equation}
    k\geq l-\frac{M(\sigma)}{2}.
    \label{eq:sm_ineq}
\end{equation}
Using this inequality, we realize Eq.~\eqref{eq:sm_even1} is at least $1/D$ smaller than Eq.~\eqref{eq:sm_even2}.

If $M(\sigma)$ is odd, we have
\begin{equation}
    \frac{\mathbb{E}_{\mathcal{E}'}[|\mathrm{tr}W|^{M(\sigma)}]}{D^{k+1-l+M(\sigma)}}\leq \frac{\mathcal{F}^{((M(\sigma)-1)/2)}_{\mathcal{E}'}}{D^{k-l+M(\sigma)}}.
    \label{eq:sm_odd1}
\end{equation}
We need to compare it with the corresponding term in Eq.~\eqref{eq:sm_leadunitary}, which is
\begin{equation}
    \frac{\mathcal{F}^{((M(\sigma)-1)/2)}_{\mathcal{E}'}}{D^{(M(\sigma)-1)/2}}.
    \label{eq:sm_odd2}
\end{equation}
Since $M(\sigma)$ is odd, the inequality Eq.~\eqref{eq:sm_ineq} implies
\begin{equation}
    k\geq l-\frac{M(\sigma)}{2}+\frac{1}{2}.
\end{equation}
Using this inequality, we realize Eq.~\eqref{eq:sm_odd1} is at least $1/D$ smaller than Eq.~\eqref{eq:sm_odd2}.

Now we go back to Eq.~\eqref{eq:sm_ushort}. For some $\pi$ and $\sigma$, $\mathrm{Wg}(\pi ^{-1}\sigma,D)\mathbb{E}_{\mathcal{E}'}[T_\sigma(W)]S_\pi$ is not real. However we can always find another pair $(\pi,\sigma)$ whose contribution is its complex conjugate, and summing them together gives twice its real part, making the frame potential a real number. For a complex number $z$, we have $|2\mathrm{Re}(z)|\leq 2|z|$. Thus by bounding their norms, we can bound the terms in Eq.~\eqref{eq:sm_ushort}. This finishes the proof of Theorem 3 for $\mathcal{E}_L$.

It remains to consider the ensemble $\mathcal{E}_R$.
In fact, after the GUE average, the frame potential of $\mathcal{E}_R$ is equal to that of $\mathcal{E}_L$. We have
\begin{equation}
\mathcal{F}^{(k)}_{\mathcal{E}_R}=\mathop{\mathbb{E}}\limits_{V,V'\sim \mathcal{E}'}k!\sum _{\vec{n}}\frac{1}{\Omega(\vec{n})}\bra{E_{n_{1}}\cdots E_{n_{k}}E_{n_{1}}\cdots E_{n_{k}}} (V'^\dagger V)^{\otimes k}\otimes(V^\dagger V')^{\otimes k}\ket{E_{n_{1}}\cdots E_{n_{k}}E_{n_{1}}\cdots E_{n_{k}}} ,
\end{equation}
which is not equal to $\mathcal{F}^{(k)}_{\mathcal{E}_L}$ in general. But if we take the GUE average, the result will be the same, since
\begin{equation}
    \begin{aligned}
        \mathrm{Tr}[V^{\dagger}_{D}(\sigma)(V'^\dagger V)^{\otimes k}\otimes(V^\dagger V')^{\otimes k}]&=\mathrm{Tr}[V^{\dagger}_{D}(\sigma) V'^{\otimes 2k}((V'^\dagger V)^{\otimes k}\otimes(V^\dagger V')^{\otimes k})V'^{\dagger \otimes 2k}]\\
        &=\mathrm{Tr}[V^{\dagger}_{D}(\sigma)(VV'^{\dagger})^{\otimes k}\otimes(V'V^{\dagger})^{\otimes k}]
    \end{aligned}
\end{equation}
where we use $[V_D^\dagger(\sigma),V'^{\otimes 2k}]=0$. This completes the proof of Theorem~3.

\end{document}